\documentclass[11pt,reqno]{amsart}
\usepackage{amsmath,amsthm,amssymb,bm,bbm,dsfont,braket,comment,cite}
\usepackage{mathtools}\mathtoolsset{centercolon}
\mathtoolsset{showonlyrefs}
\usepackage[mathscr]{eucal}
\usepackage{amsaddr}
\usepackage{upgreek}
\usepackage[left=2cm,right=2cm,top=2cm,bottom=2cm]{geometry}
\usepackage{setspace}
\usepackage{graphicx}
\usepackage{verbatim}
\usepackage{float}
\usepackage{placeins}
\usepackage{array}
\usepackage{booktabs}
\usepackage{threeparttable}
\usepackage[update,prepend]{epstopdf}
\usepackage{multirow}
\usepackage{amsfonts,amssymb,dsfont,xfrac}
\usepackage[abs]{overpic}

\usepackage[usenames,dvipsnames]{color}
\usepackage[hidelinks]{hyperref}
\hypersetup{
	unicode=false,          
	pdftoolbar=true,        
	pdfmenubar=true,        
	pdffitwindow=false,     
	pdfstartview={FitH},    
	pdftitle={My title},    
	pdfauthor={Author},     
	pdfsubject={Subject},   
	pdfcreator={Creator},   
	pdfproducer={Producer}, 
	pdfkeywords={keyword1} {key2} {key3}, 
	pdfnewwindow=true,      
	colorlinks=true,        
	linkcolor=Red,          
	citecolor=ForestGreen,  
	filecolor=Magenta,      
	urlcolor=BlueViolet,    
}
\usepackage{doi}
\usepackage{url}
\usepackage{caption, subcaption}
\usepackage[shortlabels]{enumitem}
\usepackage{shadethm}

\makeatletter
\ifx\@NODS\undefined%

\let\mathbb=\mathds
\else%
\fi
\makeatother

\def\d{{\text {\rm d}}}

\newcommand{\ccE}{\mathcal{E}}

\DeclareMathOperator{\Tr}{Tr}
\DeclareMathOperator{\e}{\mathrm{e}}

\DeclareMathOperator{\cT}{\mathcal{T}}

\DeclareMathOperator{\F}{\mathcal{F}}

\newcommand{\bE}{{\mathbf E}}

\newcommand{\be}{{\mathbf e}}

\newcommand{\tr}{\operatorname{Tr}}
\newcommand{\al}{{\alpha}}

\newcommand{\ten}{\otimes}
\newcommand{\la}{\lambda}
\newcommand{\pl}{\hspace{.1cm}}

\newcommand{\norm}[2]       {\left\lVert{{#1}}\right\lVert_{#2}} 
\newcommand{\lan}[1]{\langle #1 \vert}
\newcommand{\ran}[1]{\vert #1 \rangle }

\def\0{{\mathbf{0}}}
\def\1{{\mathbf{1}}}
\def\2{{\mathbf{2}}}
\def\3{{\mathbf{3}}}
\def\4{{\mathbf{4}}}
\def\5{{\mathbf{5}}}
\def\6{{\mathbf{6}}}

\def\7{{\mathbf{7}}}
\def\8{{\mathbf{8}}}
\def\9{{\mathbf{9}}}

\def\be{\begin{equation}}
\def\ee{\end{equation}}
\def\bea{\begin{eqnarray}}
\def\eea{\end{eqnarray}}

\def\eps{\varepsilon}

\newcommand{\id}{\operatorname{id}}

\theoremstyle{plain}
\newtheorem{theo}{Theorem} 
\newshadetheorem{conj}{Conjecture} 
\newtheorem{prop}[theo]{Proposition} 
\newtheorem{lemm}[theo]{Lemma} 
\newtheorem{coro}[theo]{Corollary} 
\newshadetheorem{shaded_theo}[theo]{Theorem} 

\theoremstyle{definition}
\newtheorem{exam}[theo]{Example}

\theoremstyle{remark}
\newtheorem{remark}[theo]{Remark}

\numberwithin{equation}{section}

\makeatletter
\newcommand{\opnorm}{\@ifstar\@opnorms\@opnorm}
\newcommand{\@opnorms}[1]{%
	\left|\mkern-1.5mu\left|\mkern-1.5mu\left|
	#1
	\right|\mkern-1.5mu\right|\mkern-1.5mu\right|
}
\newcommand{\@opnorm}[2][]{%
	\mathopen{#1|\mkern-1.5mu#1|\mkern-1.5mu#1|}
	#2
	\mathclose{#1|\mkern-1.5mu#1|\mkern-1.5mu#1|}
}
\makeatother

\begin{document}

\let\origmaketitle\maketitle
\def\maketitle{
	\begingroup
	\def\uppercasenonmath##1{} 
	\let\MakeUppercase\relax 
	\origmaketitle
	\endgroup
}

\title{\bfseries \Large{Joint State-Channel Decoupling and \\One-Shot Quantum Coding Theorem
}}

\author{  \normalsize \textsc{Hao-Chung Cheng$^{1\textrm{--}4}$, Fr{\'{e}}d{\'{e}}ric Dupuis$^{5}$, and Li Gao$^{6}$}}
\address{\small  	
	$^1$Department of Electrical Engineering and Graduate Institute of Communication Engineering,\\ National Taiwan University, Taipei 106, Taiwan (R.O.C.)\\
	$^2$Center for Quantum Science and Engineering,  National Taiwan University\\
	$^3$Hon Hai (Foxconn) Quantum Computing Center, New Taipei City 236, Taiwan (R.O.C.)\\
    $^4$Physics Division, National Center for Theoretical Sciences, Taipei 106, Taiwan (R.O.C.)\\
	$^5$D{\'e}partement d'informatique et de recherche op{\'e}rationnelle,
	Universit{\'e} de Montr{\'e}al,
	Montr{\'e}al, Qu{\'e}bec\\
	$^6$School of Mathematics and Statistics, Wuhan University, Hubei Province 430072, P.~R.~China
}

\email{\href{mailto:haochung.ch@gmail.com}{haochung.ch@gmail.com}}
\email{\href{mailto:dupuisf@iro.umontreal.ca}{dupuisf@iro.umontreal.ca}}
\email{\href{mailto:gao.li@whu.edu.cn}{gao.li@whu.edu.cn}}

\date{\today}

\begin{abstract}
In this work, we consider decoupling a bipartite quantum state via a general quantum channel. We propose a joint state-channel decoupling approach to obtain a one-shot error exponent bound without smoothing, in which trace distance is used to measure how good the decoupling is. The established exponent is expressed in terms of a sum of two sandwiched R{\'e}nyi entropies, one quantifying the amount of initial correlation between the state and environment, while the other characterizing the effectiveness of the quantum channel. This gives an explicit exponential decay of the decoupling error in the whole achievable region, which was missing in the previous results [\href{https://link.springer.com/article/10.1007/s00220-014-1990-4











        }{Commun.~Math.~Phys.~328, 2014}]. Moreover, it strengthens the error exponent  bound obtained in a recent work [\href{https://ieeexplore.ieee.org/document/10232924}{	IEEE Trans.~Inf.~Theory, 69(12), 2023}], for exponent from the channel part.
As an application, we establish a one-shot error exponent bound for quantum channel coding given by a sandwiched R\'enyi coherent information.
\end{abstract}

\maketitle

%


\section{Introduction} \label{sec:main}

Quantum decoupling is a substantial technique in quantum physics and quantum  information theory \cite{HHW+08, Dup10, DBW+14, MBD+17, LY21a}.
It concerns performing a decoupling map $\mathcal{T}_{A\to C}$ composed with a random unitary operation on a subsystem $A$ of a bipartite state $\rho_{AE}$ such that the resulting state is decoupled, i.e. ~uncorrelated with system $E$.
This task serves as a vital subroutine in achieving numerous quantum information-processing protocols such as quantum channel coding \cite{Dup10}, quantum state merging \cite{Horodecki2005}, quantum splitting \cite{Proc465}, quantum channel simulation \cite{LY21b} and so forth.

A fundamental question regarding quantum decoupling is how well can the final state be decoupled, i.e.,~close to a product from in the form of $\omega_C\ten \rho_E$, for some channel output state $\omega_C$.
Whereas a standard scenario is $\mathcal{T}$ being a partial trace channel, we in this work consider $\mathcal{T}$ to be an {\bf arbitrary quantum channel}. In this setting, the following error bound in terms of the conditional collision entropies \cite{Ren05} (also called conditional 2-entropies) was first obtained by Dupuis, Berta, Wullschleger and Renner \cite{DBW+14}.
\begin{theo}[
	Decoupling via conditional collision entropies {\cite[Theorem 3.3]{DBW+14}}] \label{theo:H2}
	For any quantum channel $\mathcal{T}_{A\to C}$ and any bipartite state $\rho_{AE}$,
	\[
	\mathbf{E}_{\mathds{U}(\mathsf{A})} \left\| \cT\left(U_A \rho_{AE} U_A^\dagger \right)- \omega_{C}\ten \rho_E\right\|_{1}
	\le \mathrm{e}^{-\frac{1}{2}H_2^*(A'{\,|\,}C)_{\omega}-\frac{1}{2}H_2^*(A{\,|\,}E)_{\rho} }.
	\]
\end{theo}

\noindent Here, $\mathbf{E}_{\mathds{U}(\mathsf{A})}$ denotes the integration over the Haar measure over the unitary group $\mathds{U}(\mathsf{A})$ on Hilbert space $\mathsf{A}$; $\omega_{A'C}$ is the Choi-Jamio{\l}kowski state of the channel $\mathcal{T}$, defined as $\omega_{A'C}:= (\mathcal{T}_{A\to C}\otimes \textnormal{id}_{A'})(\Phi_{AA'})$, where $\Phi_{AA'}$ is the maximally entangled state between system $A$ and $A'$, and $\|\cdot\|_1$ denotes the trace norm. The conditional collision entropies $H_2^*(A|B)$ is a special case of $\al$-R\'enyi sandwiched conditional entropy defined through the corresponding sandwiched relative entropy $D_\alpha^*$
\cite{MDS+13, WWY14}
\[ H_\alpha^*(A{\,|\,}B)_\rho := - \inf_{\sigma_B} D_\alpha^*( \rho_{AB}\Vert \mathds{1}_A\otimes \sigma_B)\pl, \pl D_\alpha^*(\rho\Vert \sigma) := \frac{1}{\alpha-1}\log \Tr\left[ (\sigma^{\frac{1-\alpha}{2\alpha}} \rho \sigma^{\frac{1-\alpha}{2\alpha}} )^\alpha \right]\pl.\]

This one-shot error bound of decoupling can be applied to the asymptotic scenario where the underlying state $\rho^{\ten n}$ and decoupling map $\cT^{\ten n}$ are prepared independently and identically distributed (i.i.d.). However, it does not directly give the fundamental achievability criterion of decoupling, i.e.,~positivity of the sum of conditional entropies:
\begin{align} \label{eq:first-order}
	H(A'{\,|\,}C)_{\omega} + H(A{\,|\,}E)_{\rho} > 0. \tag{1}
\end{align}
To reach the above achievability criterion, one has to invoke the \emph{smooth entropy framework}: first relaxing the conditional collision entropies estimate via the smooth conditional 2-entropies \cite[Theorem 3.1]{DBW+14},
\footnote{
A recent work by {Colomer Saus} and Winter \cite{SW23} established \eqref{eq:smooth2}
with $H_2^{\varepsilon}(A|E)_\rho $ being improved by the sandwiched R\'enyi conditional entropy $H_{\alpha}^*(A|E)_\rho$.
}
\begin{align}\label{eq:smooth2}
	\mathbf{E}_{\mathds{U}(\mathsf{A})} \left\| \cT\left( U_A \rho_{AE} U_A^\dagger  \right)- \omega_{C}\ten \rho_E\right\|_{1}
	\le \mathrm{e}^{-\frac{1}{2}H_2^\eps (A'{\,|\,}C)_{\omega}-\frac{1}{2}H_2^\eps (A{\,|\,}E)_{\rho} }+12\eps.
\end{align}
where $\eps \in (0,1)$ and $H_2^\eps(A{\,|\,}E)_{\rho}:=\sup_{\rho'\sim_\eps\rho}  H_2(A{\,|\,}E)_{\rho'}$ is the smooth conditional 2-entropy over $\eps$-ball of purified distance.
In the i.i.d. setting, the criterion \eqref{eq:first-order} can then be derived using the Quantum Asymptotic Equipartition Property (AEP) of the smooth conditional min-entropy \cite{TCR10, TH13}
\begin{align} H_2^\eps(A^n{\,|\,}E^n)_{\rho^{\ten n}}= n H(A{\,|\,}E)_{\rho}+o(n)\pl.  \label{AEP}\end{align}
Nevertheless, the resulting one-shot error exponent bound is significantly weakened.
This magnifies that the ``smoothed" one-shot error bound may not be tight in the so-called \emph{large deviation regime},  as the ``smoothing" technique is artificial and it  introduces unnecessary costs and complications.\footnote{ Large deviation regime concerns the exponential behaviors of the trace distance error, i.e.~the left-hand side of \eqref{eq:smooth2}, in the i.i.d.~asymptotic limit.
We refer the readers to Ref.~\cite{SGC24} for the discussions of the \emph{small deviation regime}, i.e.,~how large is the remainder dimension of system $C$ given a fixed trace distance error.
}

In this work, we establish the following one-shot quantum decoupling theorem using $\al$-R\'enyi sandwiched conditional entropy, which improves upon the previous bounds (e.g.,~Theorem~\ref{theo:H2}).

\begin{shaded_theo}[One-shot achievability bound] \label{theo:exp}
	For any bipartite density operator $\rho_{AE}$ and any quantum channel $\mathcal{T}_{A\to C}$ with Choi state $\omega_{A'C}:= (\mathcal{T}_{A\to C}\otimes \textnormal{id}_{A'})(\Phi_{AA'})$,
	\begin{align}
	\frac12\mathbf{E}_{\mathds{U}(\mathsf{A})} \left\|\mathcal{T}_{A\to C}\left(U_A \rho_{AE} U_A^\dagger \right) - \omega_{C}\otimes\rho_E \right\|_1
	\leq \e^{\frac{1-\alpha}{\alpha} \left( H_\alpha^*(A'{\,|\,}C)_{\omega} +H_\alpha^*(A{\,|\,}E)_\rho
 + \log 3 \right)}, \quad \forall \alpha \in [1,2].
	\end{align}
Moreover, the error exponent $\sup_{\alpha\in[1,2]} \frac{\alpha-1}{\alpha} \left( H_\alpha^*(A{\,|\,}E)_\rho + H_\alpha^*(A'{\,|\,}C)_{\omega} \right)>0$ is positive if and only if
             $ H(A'{\,|\,}C)_{\omega} + H(A{\,|\,}E)_{\rho} > 0$ is satisfied.
\end{shaded_theo}

In the i.i.d.~ scenario, our one-shot result holds for any number $n \in \mathds{N}$ of copies, and thereby it demonstrates exponential decay of trace distance for the achievability criterion \eqref{eq:first-order}, by the limit $H_\alpha^*(A{\,|\,}E)_\rho\overset{\al\to 1}{\longrightarrow} H(A{\,|\,}E)_\rho$  \emph{without smoothing}.
The proposed argument is a \emph{joint state-channel decoupling} that places an equal and joint role to the state $\rho$ and the channel $\cT$ in the decoupling. Based on that, we use complex interpolation of noncommutative $L_p$-spaces to ``R\'enyify" the contributions of the state $\rho$ and the channel $\mathcal{T}$ simultaneous.
In \cite{Dup23}, one of the present authors proved a one-shot decoupling theorem that ``R\'enyifies" only the state $\rho$ but not the channel $\mathcal{T}$, whose error exponent is in spirit as $\frac{\alpha-1}{\alpha} H_\alpha^*(A{\,|\,}E)_\rho+\frac{1}{2} H_2^*(A'{\,|\,}C)_{\omega}$. Without smoothing, this can only give exponential achievability for the channel $\cT$ being the partial traces. In contrast, our Theorem \ref{theo:exp} gives an explicit achievable error exponent for all quantum channels $\cT$ and states $\rho$.
Remarkably, the proposed joint state-channel decoupling framework works beyond independence $\omega_{A'C} \otimes \rho_{AE}$; namely, our analysis applies to the decoupling scenario that the channel could be correlated to the state.
We refer the readers to Theorem~\ref{thm:main} of Section~\ref{sec:achievability} for a more precise statement and to Section~\ref{sec:randomness-assistance} for the setting of \emph{randomness-assisted decoupling}.

Based on the joint state-channel decoupling and the pretty-good measurement method developed in \cite{CG22,shen2022strong, Cheng2022b}, we also obtain an one-shot exponential strong converse for quantum decoupling
\begin{shaded_theo}[One-shot strong converse bound]\label{theo:converse}
For any bipartite density operator $\rho_{AE}$ and for any quantum channel $\mathcal{T}_{A\to C}$ with Choi state $\omega_{A'C}$,
\begin{align}
\frac12 \bE_{\mathds{U}(\mathsf{A})}  \left\| \mathcal{T}\left( U_A \rho_{AE} U_A^\dagger \right) - \omega_C \otimes \rho_E \right\|_1 \geq 1-	2   \mathrm{e}^{ (1 - \alpha ) \left( H_{\alpha}^{\downarrow}(A' {\,|\,} C)_\omega + H_{\alpha}^{\downarrow}(A {\,|\,} E)_\rho+\log \frac{4}{3} \right) },
\quad \alpha \in (0,1).
\end{align}
In particular, the strong converse exponent $\sup_{\alpha\in (0,1) } (\alpha-1) \left( H_{\alpha}^{\downarrow}(A' {\,|\,} C)_\omega + H_{\alpha}^{\downarrow}(A {\,|\,} E)_\rho\right)>0$ is positive if and only if $ H(A'{\,|\,}C)_{\omega} + H(A{\,|\,}E)_{\rho} < 0$ is satisfied.
\end{shaded_theo}
Here, the Petz--R\'enyi conditional $\al$-entropy is defined as $H_{\alpha}^{\downarrow}( A|E )_\tau=\frac{1}{1-\alpha}\log \tr\left[ \rho_{AE}^\alpha (\mathds{1}_A\ten \rho_{E})^{1-\alpha}\right], \alpha\in (0,1)$.
By the limit $H_{\alpha}^{\downarrow}( A|E )_\rho \overset{\al\to 1}{\longrightarrow} H( A|E )_\rho$, we establish the exponential strong converse for any pair of state-channel $(\rho_A, \cT_{A\to C})$ outside the (closure of) achievable region \eqref{eq:first-order}.


As an application, we obtain the following one-shot coding theorem for entanglement-assisted quantum communication \cite{Llo97, Dev05, HHW+08, BSS+02}. We refer to Section \ref{sec:coding} for detail definition of code and error for entanglement-assisted quantum communication.
\begin{shaded_theo}[One-shot error bound for quantum communication] \label{theo:communication0}
	For any quantum channel $\mathcal{N}_{A'\to C}$ and any pure state $\sigma_{AA'}$ with $\mathsf{A}'\cong \mathsf{A}$,
	there exists a $(\mathrm{Q}, \mathrm{E}, \eps)$-entanglement-assisted quantum communication code for $\mathcal{N}_{A'\to C}$ satisfying
	\begin{align}\label{eq:error}
		\eps \leq \sqrt{ 2 \sqrt{ 2\delta_1 } + 2 \delta },
	\end{align}
	where the error terms are
	\begin{align}
		\delta_1 &= 6 \, \mathrm{e}^{ - \sup_{\alpha\in[1,2]} \frac{\alpha-1}{\alpha}\left( H_{\alpha}^*(A)_\sigma - (\mathrm{Q}+\mathrm{E}) \right) }; \\
		\delta_2 &= 6 \, \mathrm{e}^{ - \sup_{\beta \in [\sfrac23, 1]} \frac{1-\beta}{\beta} \left( I_{\beta}^*(A\rangle C)_{\mathcal{N}(\sigma)} - (\mathrm{Q} - \mathrm{E}) \right) }.
	\end{align}
	Here, $I_{\beta}^*(A\rangle B)_\rho := \inf_{\sigma_B} D_\beta^*(\rho_{AB}\Vert \mathds{1}_A\otimes \sigma_B)$ denotes the sandwiched R\'enyi coherent information
	\cite{KW20}.
	
	Moreover, the achievable error exponent (for $\eps$) is positive if and only if
	\begin{align} \label{eq:criterion}
		\mathrm{Q} + \mathrm{E} < H(A)_\sigma, \quad \text{and} \quad
		\mathrm{Q} - \mathrm{E} < I\left(A\rangle C\right)_{\mathcal{N}(\sigma)},
	\end{align}
    where $\mathrm{Q}$ stands for the number of qubits to be transmitted and $\mathrm{E}$ stands for the number of ebits to be consumed.
\end{shaded_theo}

The above coding theorem was obtained in \cite[Theorem 3.14]{dupuis2010decoupling} with smooth conditional 2-entropies, which implies the achievable rate region \eqref{eq:criterion} in the i.i.d setting by AEP \eqref{AEP}.
In contrast, our result demonstrates explicit exponential decays of the communication error for all the rate pairs $(\mathrm{Q},\mathrm{E})$ in the region \eqref{eq:criterion}, as both the R\'enyi entropy $H_{\alpha}^*(A)_\sigma$ and the sandwiched R\'enyi coherent information $I_{\beta}^*(A\rangle C)_{\mathcal{N}(\sigma)}$ are additive under product states (Proposition \ref{lemm:multiplicative}).  We emphasize that our results hold for arbitrary blocklength $n\in\mathds{N}$ without any assumption about asymptotics. It has the advantage that given a pair $(\mathrm{Q},\mathrm{E})$ and a error threshold $\eps\in (0,1)$, it is easy to find a concrete $n$ such that one can construct a $(\mathrm{Q},\mathrm{E},\eps)$-code on $n$ uses of channel $\mathcal{N}^{\ten n}$, without going through the detailed estimate in AEP \eqref{AEP}.

The rest of the paper is organized as follows: Section \ref{sec:preliminaries} reviews the definitions of the entropic quantities and related vector-valued noncommutative $L_p$ norms. The latter is the key mathematical technique in our proof of achievability. In Section 3, we prove the Theorem \ref{theo:exp} and Theorem \ref{theo:converse} through joint state-channel decoupling. Section 4 is devoted to the application to quantum coding theorem \ref{theo:communication}. We conclude the paper with a discussion on the error exponent in our results.\\

\section{Preliminaries on entropic quantity and vector-valued noncommutative $L_p$ norms} \label{sec:preliminaries}
We denote by $\mathcal{B(H)}$ the set of bounded operators on a
(possibly infinite-dimensional) separable
Hilbert space $\mathcal{H}$ and by $\mathcal{S(H)}$ the set of density operators on $\mathcal{H}$.  For any $\alpha>0$,  the Schatten $\al$-norm of an operator $X$ is defined
\begin{align} \label{eq:Schatten}
	\begin{split}
		\|X\|_{\al} &\equiv \|X\|_{S_\alpha(\mathcal{H})} := \left(\tr\left[|X|^\al\right]\right)^{\frac{1}{\al}}
	\end{split}
\end{align}
We denote $S_\alpha\left( \mathcal{H} \right) := \left\{ X \in \mathcal{B(H)} : \|X\|_\alpha < \infty \right\}$ as the Schatten $\al$-class on a Hilbert space $\mathcal{H}$.
Throughout this paper, we will use $A, B, C, \cdots$ to label quantum systems, and use $\mathsf{A}, \mathsf{B}, \mathsf{C}, \cdots$ to denote the Hilbert spaces associated to quantum systems $A, B, C, \cdots$.
We write $|A|$ as the dimension of $\mathsf{A}$.
The set of density operators on quantum system $A$ is denoted as $\mathcal{S}(\mathsf{A})$.
We use $\mathds{1}_A$ as the identity operator on $\mathsf{A}$, and $\id_{A}$ as the identity super-operator on $\mathcal{B}(\mathsf{A})$.

Let $\al\in (0,1)\cup (1,\infty]$ and $\frac{1}{\al}+\frac{1}{\al'}=1$.
For a bipartite density operator $\rho_{AB}$ on $\mathsf{A}\otimes \mathsf{B}$, the order-$\alpha$ \emph{sandwiched conditional entropy} $H_\al^*$ \cite{MDS+13 , WWY14} and
a variant \emph{Petz--R\'enyi conditional entropy} $H_\al^{\downarrow}$ \cite{Pet86} are
\begin{align}
	&H_\al^*(A {\,|\,} B)_\rho=
-\al'\log \inf_{\sigma_B \in\mathcal{S}( \mathsf{B} ) } \left\| \left(\mathds{1}_A\ten\sigma_B^{-\frac{1}{2\al'}} \right) \rho_{AB} \left(\mathds{1}_A\ten\sigma_B^{-\frac{1}{2\al'}}\right)\right\|_{\al}\pl, \label{eq:conditional_sandwiched}
\\
&H_\al^{\downarrow}(A {\,|\,} B)_\rho=
-\frac{1}{\al-1}\log \tr \left[\rho_{AB}^{\al}
 (\mathds{1}_A\ten\rho_B)^{1-\al} \right]\pl,
\label{eq:conditional_Petz}
\end{align}
where we have used the Schatten $\al$-norm defined in \eqref{eq:Schatten}.
Note that both the quantities converge to the conventional conditional entropy, i.e.,
\begin{align} \label{eq:conditional}
    \lim_{\alpha \to 1} H_\al^*(A {\,|\,} B)_\rho
    = \lim_{\alpha \to 1} H_\al^{\downarrow}(A {\,|\,} B)_\rho
    = - \Tr\left[\rho_{AB} \left(\log \rho_{AB} - \log \mathds{1}_A \otimes \rho_B \right) \right]
    =: H(A{\,|\,} B)_{\rho}.
\end{align}
Moreover, it is well known that both $H_\al^*(A {\,|\,} B)_\rho$ and $H_\al^{\downarrow}(A {\,|\,} B)_\rho$ are monotonically non-increasing in $\alpha$.

Note that the  order-$\alpha$ sandwiched conditional entropy is essentially the scaled logarithmic $S_1(\mathsf{B},S_\al(\mathsf{A}))$-norm
\begin{align}
	\norm{\rho}{S_1(\mathsf{B},S_\al(\mathsf{A}))}:=\inf_{\sigma_B \in\mathcal{S}( \mathsf{B} ) }\left\|( \mathds{1}_A\ten\sigma_B)^{-\frac{1}{2\al'}} \rho_{AB}( \mathds{1}_A\ten\sigma_B)^{-\frac{1}{2\al'}}\right\|_{\al},
\end{align}
which is a special case of a \emph{noncommutative vector-valued $L_p$-norm} introduced by Pisier \cite{pisier1998non}. For a general operator $X_{AB}\in \mathcal{B}(\mathsf{A}\otimes \mathsf{B})$, this norm is defined, for all $1\leq \alpha,r \leq \infty$ and $\frac{1}{\beta} = \left| \frac{1}{\alpha} - \frac{1}{r}\right|$, as
\begin{align}
	\norm{X_{AB}}{S_r(\mathsf{B},S_\al(\mathsf{A}))} :=
	\begin{dcases}
		\sup_{ \|a\|_{2\beta} \|b\|_{ 2\beta} \leq 1 } \left\|(\mathds{1}_A \otimes a) X_{AB} ( \mathds{1}_A \otimes b ) \right\|_{ S_\alpha (\mathsf{A}\otimes \mathsf{B}) } & \alpha \leq r; \\
		\inf_{X_{AB} =( \mathds{1}_A\ten a )Y_{AB}(\mathds{1}_A\ten b)}\norm{a}{2\beta}\norm{Y_{AB}}{\al}\norm{b}{2\beta} & \alpha \geq r. \\
	\end{dcases}
\end{align}
For the case of $\alpha \geq r = 1$,
\begin{align} \norm{X_{AB}}{S_1(\mathsf{B},S_\al(\mathsf{A}))}
	&=\inf_{X_{AB}=( \mathds{1}_A\ten a)Y_{AB}(\mathds{1}_A\ten b)}\norm{a}{2\al'}\norm{Y_{AB}}{\al}\norm{b}{2\al'}\\
	&=\inf_{\sigma_B, \,\omega_B \in\mathcal{S}( \mathsf{B} )} \left\| \left(\mathds{1}_A\ten\sigma_B^{-\frac{1}{2\al'}} \right) \rho_{AB} \left(\mathds{1}_A\ten\omega_B^{-\frac{1}{2\al'}}\right)\right\|_{\al},
\end{align}
where the first infimum is for all $a,b\in \mathcal{B}(\mathsf{B})$ with factorization $X=(1\ten a)Y(1\ten b)$ and the second infimum is for all density operators $\sigma_B$ and $\omega_B$ on $\mathsf{B}$. When $X$ is positive, it suffices to take $\sigma=\omega$ and it gives the relation
\[
H_\al^*(A {\,|\,} B)_\rho=
-\al'\log \norm{\rho}{S_1(\mathsf{B},S_\al(\mathsf{A}))}.
\]
In particular, for $\al=1$, one has
\begin{align} \label{eq:L1_reduction}
	S_1(\mathsf{B},S_1(\mathsf{A}))\cong S_1(\mathsf{A}\ten \mathsf{B})
\end{align}
being the trace class, and for $\al=\infty$ and positive $\rho_{AB}\ge 0$,
\begin{align}
\norm{\rho_{AB}}{S_1(\mathsf{B}, \,S_\infty(\mathsf{A}))}=\inf_{\sigma_B \in\mathcal{S}( \mathsf{B} ) }\{\lambda\in \mathds{R} : \rho_{AB}\le \lambda \mathds{1}_A\ten \sigma_B\}\ .
\end{align}
which corresponds to the conditional min-entropy $H_{\text{min}}(A\,|\, B)_{\rho} := - \log \norm{\rho_{AB}}{S_1(\mathsf{B}, \,S_\infty(\mathsf{A}))}$ \cite{KRS09}.
These $L_p$-spaces satisfies the duality relation
\[ S_r(\mathsf{B},S_\al (\mathsf{A}))^*=S_{r'}(\mathsf{B},S_{\al'} (\mathsf{A}))\pl, \quad  \frac{1}{\al}+\frac{1}{\al'}=\frac{1}{r}+\frac{1}{r'}=1 \]
with trace pairing $\langle X_{AB},Y_{AB} \rangle =\tr_{AB}[XY]$. In this paper, we will mostly uses $S_1(\mathsf{B}, \,S_\al(\mathsf{A}))$ and its dual space $S_\infty(\mathsf{B}, \,S_{\al'}(\mathsf{A}))$. Another useful properties about these $L_p$-spaces is that they satisfy the complex interpolation relation that for $\frac{1}{p}=\frac{1-\theta}{p_0}+\frac{\theta}{p_1}, \frac{1}{q}=\frac{1-\theta}{q_0}+\frac{\theta}{q_1}$, and $\theta\in[0,1], p_0,p_1,q_0,q_1\in [1,\infty]$,
\begin{align} S_q(\mathsf{B},S_{p}(\mathsf{A}))=[S_{q_0}(\mathsf{B},S_{p_0}(\mathsf{A})),S_{q_1}(\mathsf{B},S_{p_1}(\mathsf{A}))]_{\theta}\ ,
\end{align}
We refer to Appendix \ref{sec:interpolation} for more information about complex interpolation.

\section{Joint State-Channel Decoupling} \label{sec:joint}
\subsection{One-shot achievability} \label{sec:achievability}
In this section, we prove our main result of Theorem~\ref{theo:exp} and beyond. Let $\cT_{A\to C}:\mathcal{S}(\mathsf{A}) \to \mathcal{S}(\mathsf{C})$ be a quantum channel (i.e.,~a completely positive and trace-preserving map) and let
\[\omega_{A'C}=\id_{A'}\ten \cT_{A\to C}(\Phi_{A'A}).
\]
be its (normalized) Choi state \cite{Jam72, Cho74}. Here, we write $\Phi_{A'A}=\frac{{1}}{|A|}\sum_{i,j}\ket{i}\bra{j}_{A'}\ten \ket{i}\bra{j}_A$ as the maximally entangled state between $A$ and $A'$ and $\ket{\Phi_{A'A}}=\frac{{1}}{\sqrt{|A|}}\sum_{i}\ket{i}_{A'}\ket{i}_{A}$ as the corresponding vector.
Note that the reduced state on $\mathsf{C}$, $\omega_{C}=\cT(\frac{\mathds{1}_A}{|A|})$, is the output of maximally mixed state.
Write $d=|A|$ for short notation.
We recall that for any $X_A$ and $Y_{A'}$
\[ \bra{\Phi_{AA'}}X_A\ten Y_{A'} \ket{\Phi_{AA'}}=\frac{1}{d^2}\tr [ X^tY]\pl,\]
where $X^t$ is the transpose of $X$ with respect to the basis $\ket{i}_A$. It follows that for any density $\rho_{AE}$,
\begin{align}\label{eq:tensor}
	\cT_{A\to C}(\rho_{AE})=d^2\bra{\Phi_{AA'}}\omega_{A'C}\ten\rho_{AE}\ket{\Phi_{AA'}}.
\end{align}
Given a bipartite operator $X_{AE}$ and a unitary $U_A \in \mathds{U}(\mathsf{A})$, we write
\begin{align}
	X^U_{AE}:=(U_A\ten \mathds{1}_E)X_{AE}(U_A^\dagger\ten \mathds{1}_{E})\pl,
\end{align}
where we drop the subscript `$A$' of $U$ in $X^U_{AE}$ for simplicity. The observation \eqref{eq:tensor} leads to our key construction of joint state-channel decoupling.

For Hilbert spaces $\mathsf{C}, \mathsf{E}$ and $\mathsf{A}\simeq \mathsf{A'}$ with dimension $d = |\mathsf{A}|$, we define:
\begin{align}
	&\Theta : \mathcal{B}(\mathsf{A'}\otimes \mathsf{C}\otimes\mathsf{A}\otimes \mathsf{E}) \to L_\infty ( \mathds{U}(\mathsf{A}), \mathcal{B}(\mathsf{C}\otimes \mathsf{ E}) );\\
	&\Theta(Y_{A'ACE})(U_A) := d^2\bra{\Phi_{A'A}}  U_A \, Y_{A'ACE} \,  U_A^\dagger \ket{\Phi_{A'A}},  
\end{align}
where $L_\infty ( \mathds{U}(\mathsf{A}), \mathcal{B}(\mathsf{C}\otimes \mathsf{ E}) )$ is the space of $\mathcal{B}(\mathsf{C}\otimes \mathsf{ E})$-valued random variables on the unitary group $\mathds{U}(\mathsf{A})$. Its vector-valued $L_\alpha$ norm for $\alpha\geq 1$ is that for $X: \mathds{U}(\mathsf{A})\to \mathcal{B}(\mathsf{C}\otimes \mathsf{ E})  $
\begin{align}
\left\| X\right\|_{ L_\alpha( \mathds{U}(\mathsf{A}), S_\alpha(\mathsf{C}\otimes \mathsf{ E}) ) } := \left(
\mathbf{E}_{ \mathds{U}(\mathsf{A}) } \left\| X(U) \right\|_\alpha^\alpha \right)^{\frac{1}{\alpha}}\pl.
\end{align}
where $\mathbf{E}_{ \mathds{U}(\mathsf{A}) } $ is the integration on $\mathds{U}(\mathsf{A}) $ with respect to the Haar measure.
We remark that in the following mathematical derivations, the joint system $CE$ (associated with Hilbert space $\mathsf{C}\otimes \mathsf{E}$) can be viewed as one reference system $R$. Nevertheless, we will keep the tensor notation $\mathsf{C}\otimes \mathsf{ E}$ for the correspondence to the decoupling setting, although the tensor structure will not be used until Lemma \ref{lemm:multiplicative}.

For any $Y \in \mathcal{B}(\mathsf{A'}\otimes \mathsf{C}\otimes\mathsf{A}\otimes \mathsf{E})$, we have
\begin{align}
\mathbf{E}_{ \mathds{U}(\mathsf{A}) }\Theta(Y)= d^2 \mathbf{E}_{ \mathds{U}(\mathsf{A}) }\bra{\Phi_{A'A}}  U_AYU_A^\dagger \ket{\Phi_{A'A}}=d^2\bra{\Phi_{A'A}}  \mathrm{E}_A(Y) \ket{\Phi_{A'A}},
\end{align}
where $\mathrm{E}_A$ is the {completely depolarizing map} on system $\mathsf{A}$, i.e.
\begin{align}
&\mathrm{E}_A(X_{AB})
:=\int_{\mathds{U}(\mathsf{A})} U_AX_{AB}U_A^\dagger \, \mathrm{d}U_A
\equiv \mathbf{E}_{\mathds{U}(\mathsf{A})} \left[ U_A X_{AB} U_A^\dagger \right]
=\frac{\mathds{1}_A}{|A|}\ten \tr_A[X_{AB}]\pl, \quad X_{AB} \in \mathcal{B}(\mathsf{A}\otimes \mathsf{B}).
\end{align}
Using the above maps, we obtain the following expression for quantum decoupling
\begin{align}\label{eq:decouplemap}
\cT_{A\to C}(\rho^U_{AE})- \omega_{C}\ten \rho_E=\Theta(\omega_{A'C}\ten \rho_{AE}) (U)- \Theta\left(\omega_{A'C}\ten \mathrm{E}_A(\rho_{AE})\right)(U)=\Theta\circ (\id -\mathrm{E}_A) (\omega_{A'C}\ten \rho_{AE} ) (U)\pl.
\end{align}
From this perspective, it is natural to care about the norm of $\Theta\circ (\id - \mathrm{E}_A)$ as a map. We call $\Theta$ the decoupling map and $\Theta\circ (\id -\mathrm{E}_A)=\Theta-\mathbf{E}_{ \mathds{U}(\mathsf{A}) } \circ \Theta$ the decoupling error map.

Our first claim is the following $L_2$ lemma.
\begin{lemm}[Map norm for $L_2$]\label{lemm:L2}
	For any Hilbert space $\mathsf{C}, \mathsf{E}$ and $\mathsf{A}\simeq \mathsf{A'}$, 
	\begin{align} \norm{\Theta\circ (\id-\mathrm{E}_A): S_2(\mathsf{A'}\otimes \mathsf{C}\ten \mathsf{A}\otimes \mathsf{E})\to L_2(\mathds{U}(\mathsf{A}), S_2(\mathsf{C}\otimes \mathsf{E}))  }{}\le \frac{2}{\sqrt{3}}.
	\end{align}
	Namely, for any operator $Y_{AA'CE}$,
	\begin{align}
		\mathbf{E}_{\mathds{U}(\mathsf{A})} \left\| \Theta(Y) - \bE_{\mathds{U}{(\mathsf{A})}} \circ \Theta(Y)\right\|_{2}^2
	\le \frac{4}{3}\left\|Y\right\|_{2}^2.
	\end{align}
\end{lemm}
\begin{proof}

Write $\mathring{Y}=Y- \mathrm{E}_A(Y) \in \mathcal{B}(\mathsf{A'}\otimes \mathsf{C} \otimes \mathsf{A} \otimes \mathsf{E})$. We have for each $U\in \mathds{U}(\mathsf{A})$,
\begin{align*}
	\norm{\Theta(Y-\mathrm{E}_A(Y))(U)}{2}^2
	&= d^4\tr\Big[\Phi_{A'A}U_A \mathring{Y}^\dagger U_A^\dagger  \Phi_{A'A} U_A \mathring{Y} U_A^\dagger \Big]\\
	&= d^4\tr\Big[U_A^\dagger\Phi_{A'A}U_A \mathring{Y}^\dagger U_A^\dagger \Phi_{A'A} U_A \mathring{Y}  \Big]
    \\
	&= d^4\tr\Big[\Phi_{A'A}^{U^\dagger} \mathring{Y}^\dagger   \Phi_{A'A}^{U^\dagger} \mathring{Y} \Big]\\
	&=d^4\tr\left[F_{A'ACE} (\Phi_{A'A}^{U^\dagger}\ten\Phi_{\tilde{A}'\tilde{A}}^{U^\dagger}) (\mathring{Y}_{A'ACE}^\dagger \ten \mathring{Y}_{\tilde{A}'\tilde{A}\tilde{C}\tilde{E}}) \right]
\end{align*}
where $F_{A'ACE}$ is the swap operator between $A'ACE$ and its copy $\tilde{A}'\tilde{A}\tilde{C}\tilde{E}$. Here and in the following, we always use the ``tilde'' notation to denote copies of systems coming from the trick of swap operator (Lemma~\ref{lemm:swap} of Appendix~\ref{sec:lemmas}). Using the identity (Lemma~\ref{lemm:phi_twirling}),
\begin{align*}
	\mathbf{E}_{\mathds{U}(\mathsf{A})} \Phi_{A'A}^{U^\dagger}\ten \Phi_{\tilde{A}\tilde{A}'}^{U^\dagger}
	=\frac{1}{d^2}\left( \frac{1}{d^2-1} \mathds{1}_{AA'\tilde{A}\tilde{A}'}-\frac{1}{d^3-d}F_A\ten \mathds{1}_{A'\tilde{A}'}-\frac{1}{d^3-d} F_{A'}\ten \mathds{1}_{A\tilde{A}}+\frac{1}{d^2-1}F_{AA'}\right),
\end{align*}
we have
\begin{align*}
	\mathbf{E}_{\mathds{U}(\mathsf{A})}\norm{\Theta(\mathring{Y})(U)}{2}^2
	&=
	\frac{d^2}{d^2-1}\norm{\mathring{Y}_{A'ACE}}{2}^2-\frac{d^2}{d^3-d}
	\norm{\mathring{Y}_{A'CE}}{2}^2-\frac{d^2}{d^3-d} \norm{\mathring{Y}_{ACE}}{2}^2+\frac{d^2}{d^2-1}\norm{\mathring{Y}_{CE}}{2}^2
\end{align*}
Note that
\begin{align}
	\mathring{Y}_{A' CE}&=\tr_{A}\left[Y_{AA'CE}  -\mathrm{E}_A(Y_{AA'CE}) \right]\\
	&=\tr_{A}\circ (\id_A - \mathrm{E}_A )(Y)=0\pl.
\end{align}
We have
\begin{align*}
	\mathbf{E}_{\mathds{U}(\mathsf{A})}\norm{\Theta(\mathring{Y})(U)}{2}^2 &= \frac{d^2}{d^2-1}\norm{\mathring{Y}_{A'ACE}}{2}^2-\frac{d^2}{d^3-d} \norm{\mathring{Y}_{ACE}}{2}^2\\
	&\le\frac{d^2}{d^2-1} \left\|\mathring{Y}_{A'ACE}\right\|_{2}^2 \\
	&\le \frac{4}{3}\left\|Y_{A'ACE}\right\|_{2}^2.
\end{align*}
Here, we used the fact that
\begin{align}
	 \left\|\mathring{Y}_{A'ACE}\right\|_{2}\le \left\|Y_{A'ACE}\right\|_{2}.
\end{align}
and $\frac{d^2}{d^2-1}\le \frac{4}{3}$ for $d\ge 2$.
\end{proof}
\begin{remark}\label{rem:4of3}
It was proved in \cite[Theorem 3.3]{dupuis2014one} that for positive $Y$,
\begin{align}
		\mathbf{E}_{\mathds{U}(\mathsf{A})} \left\| \Theta(Y) - \bE_{\mathds{U}(A)}\Theta(Y)\right\|_{2}^2
	\leq \left\|Y\right\|_{2}^2.
	\end{align}
Their argument uses a fact \cite[Theorem 3.6]{dupuis2014one} that \begin{align} \label{eq:dim}\frac{1}{|A|}\norm{X_{AB}}{2}^2 \;\leq\;  \norm{X_{B}}{2}^2 \;\leq\; |A| \cdot \norm{X_{AB}}{2}^2, \end{align} where the lower bound only holds for $X \ge 0$. Indeed, an element $X_{AB}\neq 0$ but $X_{B}=0$ violates \eqref{eq:dim}. Here, for the purpose of complex interpolation, we need the $L_2$-estimate for all elements in the \textit{complex} vector space $\mathcal{B}(\mathsf{A'}\otimes \mathsf{C}\otimes\mathsf{A}\otimes \mathsf{E})$, for which the price we pay is a factor of $4/3$.
\end{remark}

Our second lemma is a simple estimate for the map norm of $\Theta\circ (\id-\mathrm{E}_A)$ between $L_1$-spaces.
\begin{lemm}[Map norm for $L_1$]\label{lemm:p=1}
	For any Hilbert space $\mathsf{C}, \mathsf{E}$ and $\mathsf{A}\simeq \mathsf{A'}$,
	\begin{align}
&\norm{\Theta: S_1(\mathsf{A'}\otimes \mathsf{C}\ten \mathsf{A}\otimes \mathsf{E})\to L_1(\mathds{U}(\mathsf{A}), S_1(\mathsf{C}\otimes \mathsf{E}))}{}\le 1,\\
&\norm{\Theta\circ (\id-\mathrm{E}_A): S_1(\mathsf{A'}\otimes \mathsf{C}\ten \mathsf{A}\otimes \mathsf{E})\to L_1(\mathds{U}(\mathsf{A}), S_1(\mathsf{C}\otimes \mathsf{E}))  }{}\le 2.
	\end{align}
	Namely, for any operator $Y_{AA'CE}$,
	\begin{align}
		\bE_{\mathds{U}(\mathsf{A})}\left\|\Theta(Y)-\bE_{\mathds{U}(\mathsf{A})} \cdot \Theta(Y) \right\|_{1}
	\le 2\left\| Y \right\|_{1}.
	\end{align}
\end{lemm}
\begin{proof} The first inequality is a consequence of that $\Theta$ is a completely positive trace-preserving map. Indeed, the complete positivity is clear from the definition and for any $Y_{A'ACE}$, we have that
\begin{align*}
\bE_{\mathds{U}(\mathsf{A})} \tr_{AA'CE}\left[ \Theta(Y)(U) \right]
&=d^2\tr_{AA'}\left[\bE_{U}(\Phi_{AA'}^U)Y_{AA'}\right]
\\
&=d^2\tr_{AA'}\left[\frac{1}{d^2} \mathds{1}_{AA'}Y_{AA'}\right]\\
&=\tr[Y].
\end{align*}
The second inequality follows from the fact that both $\id$ and $\mathrm{E}_A$ are contractions from $S_1(\mathsf{A'}\otimes \mathsf{A}\otimes \mathsf{C} \otimes \mathsf{E})$.
\end{proof}

We now apply complex interpolation to obtain our main technical theorem.

\begin{shaded_theo}\label{thm:main} For any Hilbert space $\mathsf{C}, \mathsf{E}$ $\mathsf{A}\simeq \mathsf{A'}$, and any operator $Y \in S_1(\mathsf{C}\otimes \mathsf{E}\otimes \mathsf{A}\otimes\mathsf{A'})$,
\begin{align}
\bE_{\mathds{U}(\mathsf{A})}  \left\|\Theta(Y)- \bE_{\mathds{U}(\mathsf{A})} \circ \Theta(Y)\right\|_{1}
\le 2  \cdot 3^{\frac{1-\al}{\al}} \left\|Y\right\|_{S_1(\mathsf{C}\otimes \mathsf{E},S_\al(\mathsf{A}\otimes \mathsf{A'}))},
\quad \alpha \in [1,2].
\end{align}
In particular, for any positive operator $\tau_{A'ACE}$,
\begin{align}
	\bE_{\mathds{U}(\mathsf{A})}  \left\|
	\Theta(\tau^U)- \bE_{\mathds{U}(\mathsf{A})} \circ \Theta(\tau) \right\|_{1}
	\le 2 \cdot \mathrm{e}^{\frac{1-\al}{\al}(H_\al^*(AA' {\,|\,} CE)_\tau+\log 3 )},
	\quad \alpha \in [1,2],
\end{align}
where $H_\al^*$ is defined in \eqref{eq:conditional_sandwiched}.
\end{shaded_theo}
\begin{proof}

We adopt the short notations $S_\al:=S_\al (\mathsf{C}\otimes \mathsf{E}\otimes \mathsf{A}\otimes \mathsf{A'})$ and $L_\al(S_\al):=L_\al( \mathds{U}(\mathsf{A}), S_\al(\mathsf{C}\otimes \mathsf{E}))$.
Using complex interpolation and the norm bounds established in Lemmas~\ref{lemm:p=1} and \ref{lemm:L2}, we have for $\frac{1}{\al}=\frac{1-\theta}{1}+\frac{\theta}{2}$ and $\theta=\frac{2(\al-1)}{\al}$,
\begin{align*}
\norm{\Theta\circ (\id -\mathrm{E}_A): S_\al \to L_\al(S_\al)}{}
&\le \norm{\Theta\circ (\id -\mathrm{E}_A): S_1 \to L_1(S_1)}{}^{1-\theta}\norm{\Theta\circ (\id -\mathrm{E}_A): S_2 \to L_2(S_2)}{}^{\theta}
\\
&= 2^{1-\theta}\cdot  (2/\sqrt{3})^{\theta}
= 2\cdot 3^{\frac{1-\al}{\al}}.
\end{align*}
This implies that for any $Y_{A'ACE}$ and $\alpha \in [1,2]$,
\begin{align}\label{eq:estimateal}
\left(\bE_{\mathds{U}(\mathsf{A})}  \left\|\Theta(Y)- \bE_{\mathds{U}(\mathsf{A})} \circ \Theta(Y)\right\|_{\al}^{\alpha} \right)^{\frac{1}{\al}}
\le 2 \cdot 3^{\frac{1-\al}{\al}} \left\|Y\right\|_{\al },
\end{align}
Now take $\frac{1}{\al'}+\frac{1}{\al}=1$ and $\tilde{Y}= \sigma_{CE}^{-\frac{1}{2\al'}} {Y}\gamma_{CE}^{-\frac{1}{2\al'}}$ for some arbitrary density operators $\sigma_{CE}$ and $\gamma_{CE}$ on system $CE$.
Applying H\"older's inequality and \eqref{eq:estimateal} to $\tilde{Y}$, we have
\begin{align*}
&\quad\; \bE_{\mathds{U}(\mathsf{A})}  \left\|\Theta(Y)- \bE_{\mathds{U}(\mathsf{A})} \circ \Theta(Y)\right\|_{1}
\\
&= \bE_{\mathds{U}(\mathsf{A})}  \left\|\Theta\circ (\id-\mathrm{E}_A)(Y)\right\|_{1}
\\
&\overset{\textrm{(a)}}{\le} \bE_{\mathds{U}(\mathsf{A})}  \left\| \sigma_{CE}^{\frac{1}{2\al'}} \right\|_{2\alpha'}  \left\|\sigma_{CE}^{-\frac{1}{2\al'}}\Theta\circ (\id-\mathrm{E}_A)(Y)\gamma_{CE}^{-\frac{1}{2\al'}}\right\|_{\al} \left\| \gamma_{CE}^{\frac{1}{2\al'}}\right\|_{2\alpha'}
\\
&\overset{\textrm{(b)}}{=} \bE_{\mathds{U}(\mathsf{A})}  \left\|\sigma_{CE}^{-\frac{1}{2\al'}}\Theta\circ (\id-\mathrm{E}_A)(Y)\gamma_{CE}^{-\frac{1}{2\al'}}\right\|_{\al}
\\
&\overset{\textrm{(c)}}{=} \bE_{\mathds{U}(\mathsf{A})}  \left\|\Theta\circ (\id-\mathrm{E}_A)(\tilde{Y})\right\|_{\al}
\\
&\leq \left(\bE_{\mathds{U}(\mathsf{A})}  \left\|\Theta\circ (\id-\mathrm{E}_A)(\tilde{Y})\right\|_{\al}^{\alpha} \right)^{\frac{1}{\al}}
\\
&\le 2\cdot 3^{\frac{1-\al}{\al}} \left\|\tilde{Y}\right\|_{\al }
\\
&=  2\cdot 3^{\frac{1-\al}{\al}} \left\|\sigma_{CE}^{-\frac{1}{2\al'}}Y\gamma_{CE}^{-\frac{1}{2\al'}}\right\|_{\al },
\end{align*}
where (a) follows from H\"older's inequality;
(b) is because $\sigma_{CE}$ and $\gamma_{CE}$ are densities;
in (c) we used the property $\Theta\circ (\id-\mathrm{E}_A)$ does not acting (i.e.,~an identity map) on $CE$.
Taking infimum over all densities $\sigma_{CE}$ and $\gamma_{CE}$ yields the first assertion. The second assertion follows from the definition of $H_\al^*(AA'|CE)_\tau= \al' \log \norm{ \tau }{S_1(\mathsf{C}\otimes \mathsf{E},S_\al(\mathsf{A}\otimes \mathsf{A'}))}$.
\end{proof}

\begin{remark}
In the above proof, a key property of the decoupling map $\mathring{\Theta}:=\Theta\circ (\id-\mathrm{E}_A)$ is that it only acts on $AA'$. Indeed, $\mathring{\Theta}$ can viewed as $ \mathring{\Theta}_{AA'}\ten \id_{CE}$, where
\[ \mathring{\Theta}_{AA'}:\mathcal{B}(\mathsf{A}\otimes \mathsf{A'})\to L_\infty(\mathds{U}(\mathsf{A}))\ , \ \mathring{\Theta}_{AA'}(Y_{AA'})(U)= d^2\bra{\Phi_{AA'}}Y^U\ket{\Phi_{AA'}}\pl .\]
Since $\mathsf{C}\otimes \mathsf{E}$ can be any Hilbert space, in terms of completely bounded norm $\norm{\cdot}{\textrm{cb}}$, Lemma \ref{lemm:L2} \& \ref{lemm:p=1} and Theorem \ref{thm:main} actually shows that
\[ \norm{\mathring{\Theta}_{AA'}: S_\al(\mathsf{A}\otimes \mathsf{A'})\to L_\al(\mathds{U}(\mathsf{A}))}{cb}\le 2\cdot 3^{\frac{1-\al}{\al}}\le 2\pl,\]
where the case $\al\in [1,2]$ can follow from for $\al=1,2$ by operator space interpolation. From this perspective, the argument in the proof of Theorem~\ref{thm:main} is basically Pisier's Lemma \cite[Lemma 1.7]{pisier1998non}
\begin{align*} &\norm{\mathring{\Theta}_{AA'}: S_\al(\mathsf{A}\otimes \mathsf{A'})\to L_\al(\mathds{U}(\mathsf{A}))}{\textrm{cb}}
\\ &\ge \norm{\mathring{\Theta}_{AA'}: S_\al(\mathsf{A}\otimes \mathsf{A'})\to L_1(\mathds{U}(\mathsf{A}))}{\textrm{cb}}
\\
&=  \sup_{\mathsf{R}} \norm{\id_R \ten \mathring{\Theta}_{AA'}: S_1(\mathsf{R} , S_1(\mathsf{A}\otimes \mathsf{A'}
)\to L_\al(\mathds{U}(\mathsf{A}),S_1(\mathsf{R}))}{}
\end{align*}
where the supremum is over all Hilbert space $\mathsf{R}$.
\end{remark}

To see our main Theorem \ref{theo:exp}, we need the following additivity of sandwiched R\'enyi conditional entropy that for $\al\in (0,1)\cup (1,\infty)$ and density operators $\omega\in \mathcal{S}(\mathsf{A'}\otimes \mathsf{C})$ and $\rho\in \mathcal{S}( \mathsf{A} \otimes \mathsf{E})$,
\begin{align} \label{eq:additivity_entropy}
	H_\al^*(A' A {\,|\,} C E)_{\omega \otimes \rho}
	=
	H_\al^*(A' {\,|\,} C)_\omega + H_\al^*(A {\,|\,} E)_\rho.
\end{align}
This additivity property of the sandwiched conditional R\'enyi entropy $H^*_{\alpha}$ has been proved via a duality relation \cite{MDS+13, Bei13} (see e.g.,~\cite[Corollary 5.9]{Tom16}. We provide a self-contained proof via multiplicativity of $S_1(S_\al)$-norm in Appendix \ref{sec:multi}, which holds for general operators.

As a consequence, we obtain Theorem \ref{theo:exp} as a special case of Theorem \ref{thm:main}, which we restate here.
\begin{coro}[c.f. Theorem \ref{theo:exp}]\label{cor:product}
For any density operator $\rho_{AE} \in \mathcal{S}(\mathsf{A}\otimes \mathsf{E})$ and any completely positive map $\mathcal{T}_{A\to C}$ with Choi state $\omega_{A'C}$, we have
\begin{align}
\bE_{\mathds{U}(\mathsf{A})}  \left\| \mathcal{T}\left( U_A \rho_{AE} U_A^\dagger \right) - \omega_C \otimes \rho_E \right\|_1 \leq 	2 \mathrm{e}^{\frac{1-\al}{\al} \left( H_\al^*(A' {\,|\,} C)_\omega + H_\al^*(A {\,|\,} E)_\rho+\log 3 \right) },
\quad \alpha \in [1,2].
\end{align}
In particular, the error exponent $\sup_{\alpha\in[1,2]} \frac{\alpha-1}{\alpha} \left( H_\alpha^*(A{\,|\,}E)_\rho + H_\alpha^*(A'{\,|\,}C)_{\omega} \right)>0$ is positive if and only if
             $ H_1^*(A'{\,|\,}C)_{\omega} + H_1*(A{\,|\,}E)_{\rho} > 0$
\end{coro}
\begin{proof}This is a consequence of the bound established in Theorem \ref{thm:main} and the additivity property (Proposition \ref{lemm:multiplicative}) applied to the product state $\tau_{A'ACE} = \omega_{A'C}\otimes \rho_{AE}$.
\end{proof}

\subsection{One-shot strong converse}
We shall now discuss the converse bound to Theorem \ref{thm:main}.
Given $\tau_{AA'CE} \in \mathcal{B}(\mathsf{A'\ten A\ten C\ten E})$ and a unitary $U_A\in \mathds{U}(\mathsf{A})$, we write
\[ \tau^U=U_A \tau U_{A}^\dagger.\]
Recall the notation $\ran{\Phi_{AA'}}=\frac{1}{\sqrt{d}} \sum_{i,j} \ket{i}\bra{j}\ten \ket{i}\bra{j}$ is a maximally entangled state on $\mathsf{A\ten A'}$ and $d=|\mathsf{A}|=|\mathsf{A'}|$.
Denote
\[\chi(\tau)=d^2\lan{\Phi_{AA'}} \tau\ran{\Phi_{AA'}}\pl\]
 We also write $\Phi_{AA'}=\ran{\Phi_{AA'}}\lan{\Phi_{AA'}}$ for the rank one projection.
 It is clear that
\[ \bE_{  \mathds{U}(\mathsf{A})}\chi(\tau^U)=\bE_{  \mathds{U}(\mathsf{A})}\Theta(\tau^U) =\tau_{CE}\]
In particular, if $\tau=\omega_{A'C}\ten\rho_{AE}$, with $\omega_{A'C}=\id_{A'}\ten \cT_{A\to C}(\Phi_{AA'})$ being the Choi state of a channel $\cT_{A\to C}$, then
\[ \chi(\omega_{A'C}\ten\rho_{AE}^U )=\cT_{A\to C}(\rho_{AE}^U)\ ,\;
\bE_{  \mathds{U}(\mathsf{A}) }\chi(\omega_{A'C}\ten\rho_{AE}^U) =\bE_{  \mathds{U}(\mathsf{A})}\cT_{A\to C}(\rho_{AE}^U)
\]
which is the decoupling map.
The following is a one-shot strong converse theorem for joint state-channel decoupling.

\begin{shaded_theo}\label{thm:converse}Let $\tau=\tau_{A'ACE}$ be a quantum state. Then for any $\al\in (0,1)$
\[ \frac{1}{2}\bE_{\mathds{U}(\mathsf{A})}\norm{\chi(\tau^U)-\tau_{CE} }{1}\ge 1-2 \mathrm{e}^{(\alpha - 1)\left(H_{\alpha}^{\downarrow}(AA'|CE)_\tau+\log \frac{4}{3}\right)}\pl,\]
where $H_{\alpha}^{\downarrow}$ is defined in \eqref{eq:conditional_Petz}.

\end{shaded_theo}
\begin{proof}
We use the short notation $\Phi:=\Phi_{AA'}$, $\bar{\tau}:=\tau_{CE}$, $\bE_{U}=\bE_{  \mathds{U}(\mathsf{A})}$ and $d=|\mathsf{A}|$.
Let $\{V_i\}_{i=1}^{d^2}\subset \mathds{U}(\mathsf{A})$ be an orthonormal basis (O.N.B.) of unitary operators with respect to the trace inner product. Denote
\[\ket{\Phi_i}=\mathds{1}_{A'}\ten V_i\ket{\Phi} \pl,\;
\Phi_i=(\mathds{1}\ten V_i)\Phi (\mathds{1}\ten V_i)^\dagger= \Phi^{V_i} \pl.
\]
Then $\ket{\Phi_i}$ is an O.N.B.~of maximally entangled state and $\{\Phi_i\}_{i=1}^d$ forms a projective-valued measurement. Then, for any $V_i$ in the O.N.B.,
\begin{align} \label{eq:V_i}
\chi(\tau^{V_i^\dagger U})=d^2 \bra{\Phi} V_i^\dagger U\tau U^\dagger V_i \ket{\Phi} =d^2 \bra{\Phi^{V_i }} U\tau U^\dagger \ket{\Phi^{V_i}}\pl.
\end{align}
For each $U$, we use the pretty-good measurement
\begin{align}
\Pi_U &=\frac{\chi(\tau^U)}{\chi(\tau^U)+\bar{\tau}}
\\
&:= \left(\chi(\tau^U)+\bar{\tau} \right)^{-1/2} \chi(\tau^U) \left(\chi(\tau^U)+\bar{\tau} \right)^{-1/2},
\end{align}
where $\bar{\tau} = \mathbf{E}_U \chi(\tau^U)$ is the mean.
Then, using the variational formula of trace distance \cite[Theorem 9.1]{NC09}
\begin{align*}
&\frac{1}{2}\bE_U\norm{\chi(\tau^U)-\bar{\tau} }{1}
\\
&\ge  \bE_U\tr\left[(\chi(\tau^U)-\bar{\tau})\frac{\chi(\tau^U)}{\chi(\tau^U)+\bar{\tau}}\right]
\\
&\ge  \underbrace{\bE_U\tr\left[\chi(\tau^U)\frac{\chi(\tau^U)}{\chi(\tau^U)+\bar{\tau}}\right]}_{\text{(I)}}
-\underbrace{\bE_U\tr\left[\bar{\tau}\frac{\chi(\tau^U)}{\chi(\tau^U)+\bar{\tau}}\right]}_{\text{(II)}}.
\end{align*}
Note that the second term (II) can be written as
\begin{align*}
&\bE_U\tr\left[\bar{\tau}\frac{\chi(\tau^U)}{\chi(\tau^U)+\bar{\tau}}\right]
\\
&= \bE_U\tr\left[\chi(\tau^U)\frac{\bar{\tau}}{\chi(\tau^U)+\bar{\tau}}\right]
\\
&= 1- \bE_U\tr\left[\chi(\tau^U)\frac{\chi(\tau^U)}{\chi(\tau^U)+\bar{\tau}}\right]
\\
&=1- \text{(I)}
\end{align*}
Thus, it sufficient to lower bound the first term (I).
For the first term (I),
\begin{align*}
&\bE_U\tr\left[\chi(\tau^U)\frac{\chi(\tau^U)}{\chi(\tau^U)+\bar{\tau}}\right]
\\
&\overset{\textrm{(a)}}{=} \frac{1}{d^2}\sum_{i=1}^{d^2}\bE_U\tr\left[\chi(\tau^{V_i^\dagger U})\frac{\chi(\tau^{V_i^\dagger U})}{\chi(\tau^{V_i^\dagger U})+\bar{\tau}}\right]
\\
&\overset{\textrm{(b)}}{=} \sum_{i=1}^{d^2}\bE_U\tr\left[ \bra{\Phi_i}\tau^{U} \ket{\Phi_i}\frac{d^2 \bra{\Phi_i}\tau^{U} \ket{\Phi_i}}{d^2 \bra{\Phi_i}\tau^{U} \ket{\Phi_i}+\bar{\tau}}\right]
\\
&= \bE_U\sum_{i=1}^{d^2}\tr\left[ \tau^{U} \left(\Phi_i\ten \frac{d^2 \bra{\Phi_i}\tau^{U} \ket{\Phi_i}}{d^2 \bra{\Phi_i}\tau^{U} \ket{\Phi_i}+\bar{\tau}}\right)\right].
\end{align*}
Here, (a) uses the invariance of Haar measure,
and (b) follows from \eqref{eq:V_i}.
Note that by the direct sum structure
\begin{align*}
\sum_i\Phi_i\ten \frac{d^2 \bra{\Phi_i}\tau^{U} \ket{\Phi_i}}{d^2 \bra{\Phi_i}\tau^{U} \ket{\Phi_i}+\bar{\tau}}
&=\frac{ d^2\sum_i \Phi_i\ten \bra{\Phi_i}\tau^{U} \ket{\Phi_i} }{\sum_i \Phi_i\ten( d^2 \bra{\Phi_i}\tau^{U} \ket{\Phi_i} + \bar{\tau})}
\\
&=\frac{ d^2 \ccE(\tau^U) }{d^2  \ccE(\tau^U) + \mathds{1}_{AA'}\ten \bar{\tau}},
\end{align*}
where $\ccE(\tau)= \sum_{i=1}^{d^2}\Phi_i\tau \Phi_i$ is the Pinching map onto the basis of $\{\ket{\Phi_i}\}_{i=1}^{d^2}$. Using the order inequality $d^2\ccE(\tau)\ge \tau$ \cite{Hay02}. Thus, we have
\begin{align*}
&\bE_U\tr\left[\chi(\tau^U)\frac{\chi(\tau^U)}{\chi(\tau^U)+\bar{\tau}}\right]
\\
&=   \bE_U \tr\left[ \tau^{U} \frac{ d^2 \ccE(\tau^U) }{d^2  \ccE(\tau^U) + \mathds{1}_{AA'}\ten \bar{\tau}}\right]
\\
&=  1-\bE_U \tr\left[ \tau^{U} \frac{ \mathds{1}_{AA'}\ten \bar{\tau} }{ d^2  \ccE(\tau^U) + \mathds{1}_{AA'}\ten \bar{\tau}}\right]
\\
&\ge  1-\bE_U \tr\left[ \tau^{U} \frac{ \mathds{1}_{AA'} \ten \bar{\tau}+ d^2\ccE(\tau^U)-\tau^U}{ d^2  \ccE(\tau^U) + \mathds{1}_{AA'}\ten \bar{\tau}}\right]
\\
&\overset{\textrm{(c)}}{\ge}  1-\bE_U \tr\left[ \left(\tau^{U}\right)^{\alpha} \left( \mathds{1}_{AA'}\ten \bar{\tau}+ d^2\ccE(\tau^U)-\tau^U\right)^{1-\alpha}\right]
\\
&=  1-\bE_U \tr\left[ \tau^{\alpha} \left( \mathds{1}_{AA'}\ten \bar{\tau}+ d^2\ccE(\tau^U)^{U^\dagger}-\tau\right)^{1-\alpha} \right].
\end{align*}
Here, (c) uses the trace inequality (see \cite[Lemma 3]{shen2022strong}) that for positive $X$ and $Y$,
\[ \tr[Y(X+Y)^{-\frac{1}{2}}X (X+Y)^{-\frac{1}{2}}]\le \tr\left[X^\alpha Y^{1-\alpha}\right]\pl, \pl \alpha \in (0,1)\]
Note that
\begin{align*}
\bE_U \ccE(\tau^U)^{U^\dagger }
&= \sum_{i} \bE_UU^\dagger \Phi_i U\tau U^\dagger  \Phi_iU
\\
&{=} d^2 \bE_U\Phi^{U^\dagger }\tau \Phi^{U^\dagger }
\\
&=  d^2 \bE_U \tr_{\tilde{A}\tilde{A}'} \left[ (\mathds{1}_{AA'}\ten  \tau_{\tilde{A}\tilde{A}'CE})(\Phi_{AA'}^{U^\dagger}\ten \Phi_{\tilde{A}\tilde{A}'}^{U^\dagger}\ten \mathds{1}_{CE}) \right]
\\
&=  d^2 \tr_{\tilde{A}\tilde{A}'} \left[ (\mathds{1}_{AA'}\ten  \tau_{\tilde{A}\tilde{A}'CE})(\bE_U \Phi_{AA'}^{U^\dagger}\ten \Phi_{\tilde{A}\tilde{A}'}^{U^\dagger}\ten \mathds{1}_{CE}) \right].
\end{align*}
Note that by Lemma \ref{lemm:phi_twirling},
\begin{align*}
	\mathbf{E}_{\mathds{U}(\mathsf{A})} \Phi_{A'A}^{U^\dagger}\ten \Phi_{\tilde{A}\tilde{A}'}^{U^\dagger}
	=\frac{1}{d^2}\left( \frac{1}{d^2-1} \mathds{1}_{AA'\tilde{A}\tilde{A}'}-\frac{1}{d^3-d}F_A\ten \mathds{1}_{A'\tilde{A}'}-\frac{1}{d^3-d} F_{A'}\ten \mathds{1}_{A\tilde{A}}+\frac{1}{d^2-1}F_{AA'}\right),
\end{align*}
where $F_{A}$ is the swap operator on $\mathsf{A\ten \tilde{A}}$.
Using the short notation, $F_{A}\equiv F_{A}\ten \mathds{1}_{A'\tilde{A}'}$ and $F_{A'}\equiv F_{A'}\ten \mathds{1}_{A\tilde{A}}$
\begin{align*} d^2\bE_U \ccE(\tau^U)^{U^\dagger}
&=  \tr_{\tilde{A}\tilde{A}'} \left[ (\mathds{1}_{AA'}\ten  \tau_{\tilde{A}\tilde{A}'CE})\left(\frac{1}{d^2-1} \mathds{1}+ \frac{1}{d^2-1} F_{AA'}-\frac{1}{d(d^2-1)} F_{A}- \frac{1}{d(d^2-1)} F_{A'}\right) \right]
\\
&= \frac{1}{d^2-1} \mathds{1}_{AA'}\ten \tau_{CE}+ \frac{1}{d^2-1} \tau_{AA'CE} -\frac{1}{(d^2-1)d} \mathds{1}_{A'}\ten \tau_{ACE}- \frac{1}{(d^2-1)d} \mathds{1}_{A}\ten \tau_{A'CE}
\end{align*}
Then we have
\begin{align*} &\bE_U \left[ \mathds{1}_{AA'}\ten \bar{\tau}+ d^2\ccE(\tau^U)^{U^\dagger}-\tau\right]
\\
&= \left(1+\frac{1}{d^2-1 }\right) \mathds{1}_{A'A}\ten \tau_{CE}-\frac{1}{d(d^2-1)} \mathds{1}_{A'}\ten \tau_{ACE}- \frac{1}{d(d^2-1)} \mathds{1}_{A}\ten \tau_{A'CE}-\frac{2-d^2}{d^2-1} \tau
\\
&\le \frac{4}{3} \mathds{1}_{AA'}\ten \tau_{CE},
\end{align*}
where we have used the fact $\frac{1}{d^2-1}\le \frac{1}{3}$ for $d\ge 2$. Putting this back to our estimate, we have
\begin{align*}
\frac{1}{2}\bE_U\norm{\chi(\tau^U)-\bar{\tau} }{1}
&\ge   1-2\bE_U \tr\left[ \tau^{\alpha} \left( \mathds{1}_{AA'} \ten \bar{\tau}+ d^2\ccE(\tau^U)^{U^*}-\tau\right)^{1-\alpha} \right]
\\
&\ge   1-2\cdot \left(\frac{4}{3}\right)^{1-\alpha} \tr\left[ \tau^{\alpha} ( \mathds{1}_{AA'}\ten \tau_{CE})^{1-\alpha} \right]
\\
&=   1-2 \mathrm{e}^{ (1-\al) \left(H_{\alpha}^{\downarrow}(AA'|CE)_\tau + \log \frac{4}{3}\right)},
\end{align*}
That finishes the proof.
\end{proof}

It is clear from the definition of $H_{\alpha}^{\downarrow}$ that for product density $\tau_{A'ACE}=\omega_{A'C}\ten \rho_{AE}$,
\begin{align*}&H_{\alpha}^{\downarrow}(AA'|CE)_{\omega\ten \rho}=H_{\alpha}^{\downarrow}(A'|C)_{\omega}+H_{\alpha}^{\downarrow}(A|E)_{\rho}\qedhere \end{align*}
Then Theorem \ref{theo:converse} follows from this special case as a corollary, which we restate here.

\begin{coro}[c.f.~Theorem \ref{theo:converse}]
For any density operator $\rho_{AE} \in \mathcal{S}(\mathsf{A}\otimes \mathsf{E})$ and for any quantum channel $\mathcal{T}_{A\to C}$ with Choi state $\omega_{A'C}$,
\begin{align}
\frac12\bE_{\mathds{U}(\mathsf{A})}  \left\| \mathcal{T}\left( U_A \rho_{AE} U_A^\dagger \right) - \omega_C \otimes \rho_E \right\|_1 \geq 1-	2   \mathrm{e}^{ (1-\alpha) \left( H_{\alpha}^{\downarrow}(A' {\,|\,} C)_\omega + H_{\alpha}^{\downarrow}(A {\,|\,} E)_\rho + \log \frac{4}{3} \right) },
\quad \alpha \in (0,1).
\end{align}
In particular, the strong converse exponent $\sup_{\alpha\in(0,1)} (\alpha-1) \left( H_{\alpha}^{\downarrow}(A' {\,|\,} C)_\omega + H_{\alpha}^{\downarrow}(A {\,|\,} E)_\rho\right)>0$ is positive if and only if $ H(A'{\,|\,}C)_{\omega} + H(A{\,|\,}E)_{\rho} < 0$.
\end{coro}

\begin{exam}
  {\rm For $\mathsf{A}=\mathsf{A}_1\otimes \mathsf{C}$ and $\cT=\id_{C}\ten \tr_{A_1}$ being the partial trace,
  \[\omega_{A'C}=\frac{\mathds{1}_{A_1'}}{|{A_1}|}\ten \Phi_{C'C}\]
  and the R\'enyi conditional entropy is
\begin{align*}
H_{\alpha}^{\downarrow}(A'|C)_{\omega}
&=H_{\alpha}^{\downarrow}(A_1'C'|C)_{\omega}
\\
&=H_{\alpha}(A_1')_{ \frac{\mathds{1}}{|{A_1}|}}+H_{\alpha}^{\downarrow}(C'|C)_{\Phi_{C'C}}
\\
&=\log |A_1|-\log |C|=\log \frac{|A|}{|C|}-\log |C|=\log \frac{|A|}{|C|^2}.
\end{align*}
}
\end{exam}

\begin{remark}\label{rem:sum}{\rm
Note that our converse exponent
\[H_{\alpha}^{\downarrow}(A'A|CE)_\tau\to H(A'A|CE)_\tau\pl,\]
 when $\al\to 1$. On the other hand, our achievability exponent is
 \[ H_\al^*(AA' {\,|\,} CE)_\tau\to H(A'A|CE)_\tau\pl, \]
 as $\al\to 1$.
In the case of $\tau=\omega_{A'C}\ten \rho_{CE}$,
\[ H(A'A|CE)_\tau=H(A'|C)_{\omega}+H(A|E)_{\rho}\pl,\]
This sum expression represents the counter forces in decoupling: if the state $\rho_{AE}$ is more entangled, then it is harder to be decoupled; if the channel $A_1$ is more entanglement-breaking, than it is easier to decouple a state $\rho_{AE}$. Two extreme examples: if $\rho_{AE}$ is a maximally entangled state and $\cT=\id_A$ is the perfect channel, then
\[\frac12\left\|\rho_{AE}^U-\frac{1}{d}\mathds{1}_{A}\ten \rho_E\right\|_{1}=\frac{d^2-1}{d^2}
\] for each $U$;
if $\cT=\tr_A$ is the fully depolarizing channel, the decoupling error is always zero.
From this perspective, our Theorem~\ref{thm:main} and Theorem \ref{thm:converse}~together gives a clear asymptotic achievable region for quantum decoupling, and strong converse outside the closure of it. That is, in the i.i.d.~asymptotic setting, the decoupling error of the channel $\cT^{\ten n}$ and the state $\rho_{AE}^{\ten n}$ using the twirling from $\mathds{U}(\mathsf{A}^n)$ will
\begin{enumerate}\item[(i)] decay exponentially to $0$ if
$H(A'|C)_{\omega}+H(A|E)_{\rho}>0$.
\item[(ii)] converges exponentially to $1$ if
$H(A'|C)_{\omega}+H(A|E)_{\rho}<0$.
It is worth emphasizing that both Theorem~\ref{thm:main} and Theorem \ref{thm:converse} hold for arbitrary copies $n$.
\end{enumerate}}
\end{remark}

\subsection{Randomness-assisted quantum decoupling} \label{sec:randomness-assistance}
A crucial conceptual tool in previous subsections is the joint state-channel decoupling map, in which the state and the channel take a joint and symmetric role.
The achievability and strong converse theorem of the usual quantum decoupling follows as a special case of an tensor product joint state $\omega_{A'C}\otimes \rho_{AE}$. This, in particular, generalizes one of our author's previous result in which the conditional entropy term of the Choi state $\omega_{A'C}$ has not been "R\'enyified".
Furthermore, our  joint state-channel decoupling approach allows us to handle a more general scenario where the underlying state $\rho$ may be correlated with the decoupling channel $\cT$.

Suppose that the state to be decoupled and the decoupling channel are \emph{statistically correlated}. That is, the system $AE$ is prepared as a random bipartite $\rho_{AE} = \sum_i p(i) \rho_{AE}^{(i)}$ for some probability distribution $p$, and moreover the randomness is shared to Alice such that for each preparation $\rho_{AE}^{(i)}$, Alice chooses to apply a decoupling map $\mathcal{T}^{(i)}$ accordingly. We term this setting as \textbf{randomness-assisted quantum decoupling} (see Figure~\ref{fig:decoupling} below).
Now, the question, again, is how well the system $A$ can be decoupled from $E$ in this scenario?

This scenario is the special case of Theorem \ref{thm:main} and \ref{thm:converse} applying to separable joint states.

\begin{coro}[One-shot randomness-assisted decoupling] \label{theo:joint}
	Let $\left\{p(i), \rho_{AE}^{(i)}, \mathcal{T}_{A\to C}^{(i)} \right\}_i$ be an ensemble of state-channel pairs. For any $\al\in[1,2]$ and $\beta \in (0,1)$
\begin{align}\label{eq:random}
1-2\mathrm{e}^{(1-\beta) (H_{\beta}^{\downarrow}(AA'|CE)_\tau-\log \frac{4}{3})}
&\le \frac{1}{2}\bE_{\mathds{U}(\mathsf{A})} \left\| \sum_i p(i)\mathcal{T}_{A\to C}^{(i)}\left(U_A \rho_{AE}^{(i)} U_A^\dagger \right) - \omega_{C}^{(i)} \otimes\rho_E^{(i)} \right\|_1 \nonumber
\\
&\le  \mathrm{e}^{ \frac{1-\alpha}{\alpha} (H_\al^*(AA' {\,|\,} CE)_\tau+\log 3 )},
\end{align}
where $\tau_{A'ACE}=\sum_{i}p(i)\omega_{A'C}^{(i)}\ten \rho_{AE}^{(i)}$ and $\omega_{A'C}^{(i)}$ is the Choi state of $\mathcal{T}^{(i)}$.
\end{coro}

\begin{remark}{
A naive approach to the upper bound (i.e., achievability) above is to apply the triangle inequality for the trace norm to bound the decoupling error for each realization $i$ separately. Then, the trace distance will be dominated by exponential decay with the smallest error exponent:
\begin{align}
	\min_i \sup_{ \alpha \in [1,2]} \frac{\alpha-1}{\alpha} \left(  H_\alpha^*(A{\,|\,}E)_{\rho^{(i)}} + H_\alpha^*(A'{\,|\,}C)_{{\omega}^{(i)}} \right),
\end{align}
Asymptotically, in the i.i.d.~scenario where $\rho_{AE}^{(i)} \leftarrow (\rho_{AE}^{(i)})^{\otimes n}$
and
$\omega_{A'C}^{(i)} \leftarrow (\omega_{A'C}^{(i)})^{\otimes n}$,
 the decoupling is achievable if
\begin{align}
 \min_i \left\{	H(A{\,|\,}E)_{\rho^{(i)}} + H(A'{\,|\,}C)_{{\omega}^{(i)}} \right\} >0.
\end{align}
However, the Corollary \eqref{theo:joint} allows to harness the shared randomness to accomplish a much improved achievability criterion in terms of one single conditional entropy:
\begin{align}
	H(A'A{\,|\,}CE)_{\tau} >0,
\end{align}
where the joint state is $\tau_{A'ACE} := \sum_i p(i) \cdot \omega_{A' C}^{(i)} \otimes \rho_{AE}^{(i)}$.}
\end{remark}

\begin{figure}[h!]
	\centering
	\includegraphics[width = 0.8\linewidth]{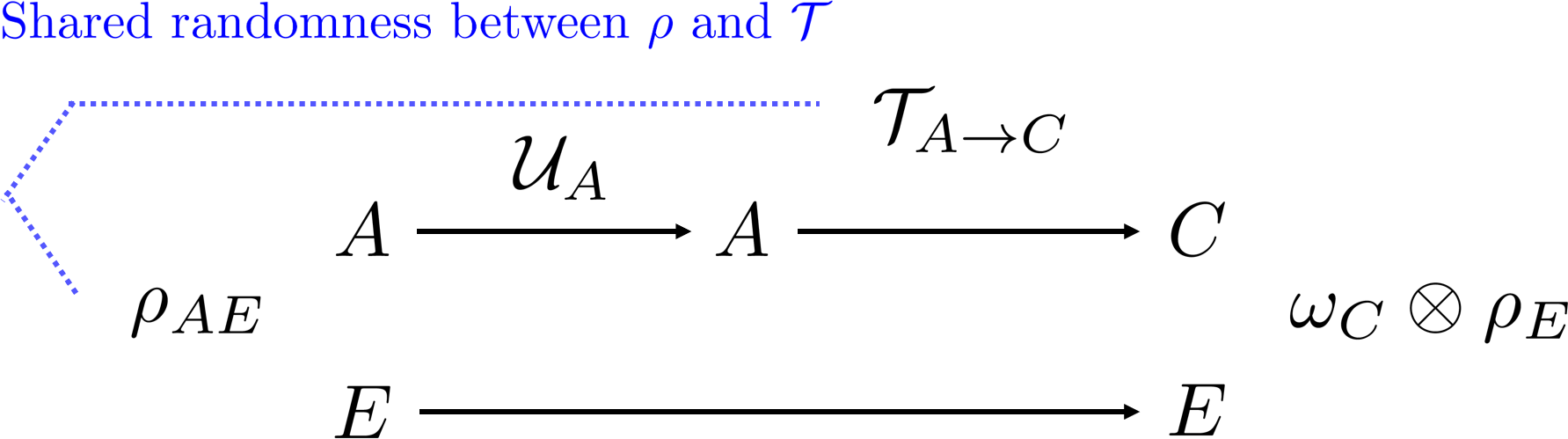}
	\caption {An illustration of randomness-assisted decoupling.}
	\label{fig:decoupling}
\end{figure}

\section{Applications: A One-Shot Quantum Coding Theorem} \label{sec:coding}
Our one shot quantum-decoupling Theorem~\ref{theo:exp} has an immediate application in quantum channel coding.
we consider the following problem:
Alice and Bob share a pure state $\psi_{ABR}$, where Alice holds $A$, Bob holds $B$, and $R$ is a reference system that purifies the state.
Alice aims to send her part of the state to Bob through a single use of the quantum channel $\mathcal{N}_{A'\to C}$.

The following theorem is a variation of \cite[Theorem 3.14]{Dup10} without smoothing.
\begin{shaded_theo}[One-shot quantum coding theorem] \label{theo:coding}
	Let $\psi_{ABR}$ be a pure state, $\mathcal{N}_{A'\to C}$ be any completely positive trace-preserving map with Stinespring dilation $\mathcal{U}_{A'\to CE}^{\mathcal{N}}$ and complementary channel $\bar{\mathcal{N}}_{A'\to E}$, let $\omega_{A''CE} = \mathcal{U}^{\mathcal{N}}\left(\sigma_{A''A'}\right)$, where $\sigma$ is any pure state  with $\mathsf{A}''\cong \mathsf{A}'$.
	Then, there exists an encoding partial isometry $\mathcal{V}_{A\to A'}$ and a decoding map $\mathcal{D}_{CB\to AB}$ such that
	\begin{align}\label{eq:coding}
		\begin{dcases}
			\frac12\left\| \bar{\mathcal{N}}\circ \mathcal{V}(\psi_{AR}) - \omega_E \otimes \psi_R \right\|_1 \leq \sqrt{2 \delta_1 } + \delta_2, \\
			\frac12\left\| \mathcal{D}\circ{\mathcal{N}}\circ \mathcal{V}(\psi_{ABR}) - \psi_{ABR} \right\|_1 \leq \sqrt{2\sqrt{2 \delta_1 } + 2\delta_2},
		\end{dcases}
	\end{align}
	where the error terms are
	\begin{align}
		\delta_1 &= 6 \, \mathrm{e}^{ - \sup_{\alpha\in[1,2]} \frac{\alpha-1}{\alpha}\left( H_{\alpha}^*(A'')_\omega - H^*_{\frac{\alpha}{2\alpha-1}}(A)_\psi\right)   }; \\
		\delta_2 &= 6 \, \mathrm{e}^{ - \sup_{\beta\in[1,2]} \frac{\beta-1}{\beta}\left( H_{\beta}^*(A''{\,|\,}E)_\omega + H^*_{\beta}(A{\,|\,}R)_\psi \right)   }. 	\end{align}
\end{shaded_theo}
\begin{proof}
We follow the argument of \cite[Theorem 3.14]{Dup10} using our improved one-shot decoupling Theorem~\ref{theo:exp}.
Let $W_{A\to A'}$ be a full rank isometry and define two CP map
\begin{align*}
&\mathcal{E}_{A'\to \mathbb{C}}(\cdot )=d_{A'} \tr_{A'}( \sigma_{A'}^{1/2}\cdot  \sigma_{A'}^{1/2})\pl, \\
&\mathcal{T}_{A'\to E}(\cdot  )= d_{A'} \bar{\mathcal{N}}( \sigma_{A'}^{1/2} \cdot  \sigma_{A'}^{1/2})\pl.
\end{align*}
Applying Theorem \ref{theo:exp} to $(\cT, W\psi_{AR}W^\dagger)$ and $(\tr, W\psi_{ABR}W^\dagger)$ respectively, we have the following deviation inequality
\begin{align}\label{eq:2}
\bE_{\mathds{U}(\mathsf{A})}\norm{d_{A'}\bar{\mathcal{N}}\left( \sigma_{A'}^{1/2} U (W\psi_{AR}W^\dagger)  U^\dagger  \sigma_{A'}^{1/2}\right)-\omega_E\ten \psi_R}{1}
&\le 2\cdot \mathrm{e}^{ \frac{\beta-1}{\beta}\left( H_\beta^*(A''{\,|\,}E)_\omega + H^*_{\beta}(A{\,|\,}R)_\psi \right)}
\\
\bE_{\mathds{U}(\mathsf{A})}\norm{d_{A'}\tr\left( \sigma_{A'}^{1/2} U (W\psi_{ABR}W^\dagger)  U^\dagger  \sigma_{A'}^{1/2}\right)-\psi_{BR}}{1}
&\le 2\cdot \mathrm{e}^{  \frac{{\alpha}-1}{{\alpha}}\left( H_{\alpha}^*(A'')_\omega + H^*_{{\alpha}}(A{\,|\,}BR)_\psi \right)}
\\
&\le 2\cdot \mathrm{e}^{  \frac{{\alpha}-1}{{\alpha}}\left( H_{\alpha}^*(A'')_\omega - H^*_{\frac{{\alpha}}{2{\alpha}-1}}(A)_\psi \right)}
\end{align}
where we use the duality of conditional $H_\al^*$ \cite{Bei13, MDS+13} that for a pure state $\phi_{ABC}$,
\[H_\al^*(A|B)_\phi=-H_{\frac{\al}{2\al-1}}^*(A|C)_\phi\pl, \pl \al\in [\sfrac{1}{2},\infty]\pl. \]
By Markov inequality, there exists a unitary $U\in \mathds{U}(\mathsf{A})$ such that
\begin{align}\label{eq:3}
\norm{d_{A'}\bar{\mathcal{N}}\left( \sigma_{A'}^{1/2} U (W\psi_{AR}W^\dagger)  U^\dagger  \sigma_{A'}^{1/2}\right)-\omega_E\ten \psi_R}{1}
&\le 6\cdot \mathrm{e}^{ \frac{\beta-1}{\beta}\left( H_\beta^*(A''{\,|\,}E)_\omega + H^*_{\beta}(A{\,|\,}R)_\psi \right)}
\\
\norm{d_{A'}\tr\left( \sigma_{A'}^{1/2} U (W\psi_{ABR}W^\dagger)  U^\dagger  \sigma_{A'}^{1/2}\right)-\psi_{BR}}{1}
&\le 6\cdot \mathrm{e}^{  \frac{\alpha-1}{\alpha}\left( H_\alpha^*(A'')_\omega - H^*_{\frac{\alpha}{2\alpha-1}}(A)_\psi \right)}\nonumber
\end{align}
By Uhlmann's Lemma (c.f. \cite[Theorem 3.1]{Dup10}), the second inequality implies that is an encoding isometry $V^{A\to A'}$ such that
\begin{align*}&\norm{d_{A'}\sigma_{A'}^{1/2} U (W\psi_{ABR}W^\dagger)  U^\dagger  \sigma_{A'}^{1/2}-V\psi_{ABR}V^\dagger}{1}\le 2\sqrt{6}\cdot \mathrm{e}^{  \frac{\alpha-1}{2\alpha}\left( H_\alpha^*(A'')_\omega - H^*_{\frac{\alpha}{2\alpha-1}}(A)_\psi \right)}=2\sqrt{\delta_1}
\end{align*}
By triangle inequality, this with \eqref{eq:3} implies
\begin{align*}&\norm{\bar{\mathcal{N}}(V\psi_{AR}V^\dagger)-\omega_E\ten \psi_{R}}{1}\le 2\sqrt{\delta_1}+\delta_2
\end{align*}
Using Ulhmann's Lemma again, we obtain a decoding partial isometry $D_{CB\to FAB}$ such that
\begin{align*}&\norm{D U_{\bar{\mathcal{N}}}(V\psi_{ABR}V^\dagger)U_{\bar{\mathcal{N}}}^\dagger D^\dagger-\xi_{EF}\ten \psi_{ABR}}{1}\le 2\sqrt{2\sqrt{\delta_1}+\delta_2}
\end{align*}
for some state $\xi_{EF}$. The second assertion in \eqref{eq:coding} follows from tracing over system $EF$.
\end{proof}
Theorem~\ref{theo:coding} leads to an one-shot error exponent for communicating quantum information with or without some entanglement assistance.
We illustrate the quantum communication task here.
Let the state $\psi_{ABR} = \Phi_{RM} \otimes \Phi_{\widetilde{A} B}$,
$M$ denotes the quantum system at Alice to be sent to Bob, and $\Phi_{\widetilde{A} B}$ is the shared entanglement between Alice and Bob.
We say that an encoder-decoder pair $(\mathcal{E}_{M \widetilde{A}\to A'}, \mathcal{D}_{C B \to M})$ is a $(\mathrm{Q}, \mathrm{E}, \eps)$-code for $\mathcal{N}_{A'\to C}$ if $|\mathsf{M}| = |\mathsf{R}|  \ge  \mathrm{2}^\mathrm{Q}$, $ |\tilde{\mathsf{A}}| =  |\mathsf{B}| \le \mathrm{2}^{\mathrm{E}}$, and
\begin{align}
	\frac12 \left\| \mathcal{D}\circ\mathcal{N}\circ \mathcal{E} \left( \Phi_{RM} \otimes \Phi_{\widetilde{A} B } \right) - \Phi_{RM} \right\|_1 \leq \eps.
\end{align}
A rate pair $(Q,E)$ is called achievable if for any $\eps>0$, there exists some $n\in \mathbb{N}^+$ and a $(n\mathrm{Q}, n\mathrm{E}, \eps)$-code for $\mathcal{N}^c$. It was proved in \cite[Theorem 3.15]{Dup10} that the following region is achievable
\begin{align}
		\mathrm{Q} + \mathrm{E} < H(A)_\sigma, \quad \text{and} \quad
		\mathrm{Q} - \mathrm{E} < I\left(A\rangle C\right)_{\mathcal{N}(\sigma)}.
	\end{align}
for arbitrary bipartite channel input $\sigma_{A'A}$.

Here, we obtain the following \emph{one-shot} achievable error exponent for quantum communication, which demonstrates exponential decays of the communication error for all rate pairs $(\mathrm{Q},\mathrm{E})$ in the achievable region \cite{Llo97, Dev05, HHW+08, BSS+02}.
\begin{shaded_theo}[One-shot error exponent for quantum communication] \label{theo:communication}
	For any quantum channel $\mathcal{N}_{A'\to C}$ and any pure state $\sigma_{AA'}$ with $\mathsf{A}'\cong \mathsf{A}$,
	there exists a $(\mathrm{Q}, \mathrm{E}, \eps)$-code for $\mathcal{N}_{A'\to C}$ satisfying
	\begin{align}
		\eps \leq \sqrt{ 2 \sqrt{ 2\delta_1 } + 2 \delta_2 },
	\end{align}
	where the error terms are
	\begin{align}\label{eq:4}
		\delta_1 &= 6 \, \mathrm{e}^{ - \sup_{\alpha\in[1,2]} \frac{\alpha-1}{\alpha}\left( H_{\alpha}^*(A)_{\sigma} - (\mathrm{Q}+\mathrm{E}) \right) }; \\
		\delta_2 &= 6 \, \mathrm{e}^{ - \sup_{\beta \in [\sfrac23, 1]} \frac{1-\beta}{\beta} \left( I_{\beta}^*(A\rangle C)_{\mathcal{N}(\sigma)} - (\mathrm{Q} - \mathrm{E}) \right) }.
	\end{align}
	Here, $I_{\beta}^*(A\rangle B)_\rho := \inf_{\sigma_B} D_\beta^*(\rho_{AB}\Vert \mathds{1}_A\otimes \sigma_B)=-H_\al^*(A|B)_\rho$ denotes the sandwiched R\'enyi coherent information
	\cite{KW20}.
	
	Moreover, the error exponent (for $\eps$) is positive if and only if
	\begin{align} 
		\mathrm{Q} + \mathrm{E} < H(A)_\sigma, \quad \text{and} \quad
		\mathrm{Q} - \mathrm{E} < I\left(A\rangle C\right)_{\mathcal{N}(\sigma)}.
	\end{align}
\end{shaded_theo}
\begin{proof}
	Let $\psi_{M\tilde{A}BR}=\Phi_{RM}\ten \Phi_{\tilde{A}B}$ be the maximally entangled state with $|R|=|M|=2^R$ and $|\tilde{A}|=|B|=2^{\mathrm{E}}$. Let $U_{\mathcal{N}}$ be the Stinespring dilation of channel $\mathcal{N}$. Applying Theorem~\ref{theo:coding} with the input state $\psi$, the channel $\mathcal{N}$, and $\omega_{ACE}=U_{\mathcal{N}}\sigma U_{\mathcal{N}}^\dagger$, we obtain an isometry $V_{M\tilde{A}\to A'}$ and decoder map $\mathcal{D}_{CB\to M}$ such that
\[ \norm{\mathcal{D}\circ \mathcal{N} (V \psi_{M\tilde{A}BR}V^\dagger)- \psi_{MR}}{1}\le 2\sqrt{2\sqrt{\delta_1}+\delta_2}\pl, \]
where
\begin{align}
		\delta_1 &= 6 \, \mathrm{e}^{ - \sup_{\alpha\in[1,2]} \frac{\alpha-1}{\alpha}\left( H_\alpha^*(A'')_\omega - H^*_{\frac{\alpha}{2\alpha-1}}(\tilde{A}M)_\psi\right)   }; \\
		\delta_2 &= 6 \, \mathrm{e}^{ - \sup_{\alpha\in[1,2]} \frac{\alpha-1}{\alpha}\left( H_\alpha^*(A''{\,|\,}E)_\omega + H^*_{\alpha}(\tilde{A}M{\,|\,}R)_\psi \right)   }.
	\end{align}
as $A=\tilde{A}M$.
Since $\psi=\Phi_{RM}\ten \Phi_{\tilde{A}B}$ is a maximally entangled statement,
we have
\begin{align*} &H^*_{\frac{\alpha}{2\alpha-1}}(\tilde{A}M)_\psi=H^*_{\frac{\alpha}{2\alpha-1}}\left(\frac{1}{d_{\tilde{A}M}} \mathds{1}_{{{\tilde{A}M}}} \right)
=\log|\tilde{A}|\cdot |M|=\mathrm{Q}+\mathrm{E},\\
& H^*_{\alpha}(\tilde{A}M{\,|\,}R)_\psi= H^*_{\alpha}(\tilde{A})_\psi + H^*_{\alpha}(M{\,|\,}R)_\psi=\mathrm{Q}-\mathrm{E},
\end{align*}
where for the second term we used Proposition \ref{lemm:multiplicative} for product state $\psi_{\tilde{A}MR}=\frac{1}{d_{\tilde{A}}}\mathds{1}_{\tilde{A}}\ten \Phi_{MR}$. The assertion follows from the duality of conditional $H_\al^*$ that
\[  H_\alpha^*(A''{\,|\,}C)_\omega=-H_{\beta}^*(A''{\,|\,}C)_\omega=I_{\beta}^*(A\rangle C)_{\mathcal{N}(\sigma)} \pl,\]
for $\beta=\frac{\al}{2\al-1}$. When $\al\in [1,2]$, $\beta\in [\frac{2}{3},1]$. That finished the proof.
\end{proof}
The first error term $\delta_1$ characterizes how well both the quantum information $\mathrm{Q}$ and Alice's share of entanglement $\mathrm{E}$ fits into the channel input and the resulting error exponent. The second error term $\delta_2$ corresponds how well the channel can be used to communicate quantum information without using entanglement, and we show that this is characterized by the sandwiched R\'enyi coherent information.
When $E=0$, this gives an one-shot error exponent for quantum communication \emph{without} assistance of entanglement. On the other hand, if we do not limit the rate of entanglement assistance, our result implies that as long as $\mathrm{Q} < H(A)_\sigma + I(A\rangle C)_{\mathcal{N}(\sigma)} = \frac12 I(A{\,:\,}C)_{\mathcal{N}(\sigma)}$, we can choose a proper $\mathrm{E}$ such that both the exponents of $\delta_1$ and $\delta_2$ are positive, hence recovers the direct coding of entanglement-assisted quantum communication.

\section{Discussion on Interpretation of the error exponent} \label{sec:interpretation}
In this section, we discuss the error exponent obtained for quantum decoupling.
Because $\alpha \mapsto H_\alpha^*$ is continuously decreasing, the positivity of the error exponent in Theorem~\ref{theo:exp} can be characterized as follows:
\begin{align}
	\sup_{ \alpha \in [1,2]} \frac{\alpha-1}{\alpha} \left( H_\alpha^*(A{\,|\,}E)_\rho + H_\alpha^*(A'{\,|\,}C)_{\omega} \right) > 0
	\quad \Longleftrightarrow \quad  H_1^*(A{\,|\,}E)_\rho + H_1^*(A'{\,|\,}C)_{\omega} >0.
\end{align}
The above exponent can be interpreted as a trade-off or balance between two error exponents associated with different tasks by using Fenchel's duality, which we illustrate below. For any bipartite density operator $\rho_{AB}$, we define the following error exponent functions:
\begin{align}
	{\color{BrickRed}
		E^*_\textnormal{\texttt{mother}}(R)_{\rho_{AE}}
	}
	&:= \sup_{\alpha\in[1,2]} \frac{\alpha-1}{\alpha} \left( H_\alpha^*(A{\,|\,}E)_\rho + \log |A| - 2R \right); \\
	{\color{RoyalBlue}
		E^*_\textnormal{\texttt{father}}(R)_{\omega_{A' C}}
	}
	&:= \sup_{\alpha\in[1,2]} \frac{\alpha-1}{\alpha} \left( 2R - I_\alpha^*(A'{\,:\,}C)_\omega\right),
\end{align}
where the order-$\alpha$ sandwiched R\'enyi conditional entropy $H_\alpha^*(A{\,|\,}B)_\rho$ and the sandwiched R\'enyi information are defined as
\begin{align}
	H_\alpha^*(A{\,|\,}E)_\rho := - \inf_{\sigma_E \in \mathcal{S}(\mathsf{E})} D_\alpha^*\left( \rho_{AE} \Vert \mathds{1}_{A} \otimes \sigma_E \right)\pl, \pl
	I_\alpha^*(A'{\,:\,}C)_\rho := \inf_{\sigma_C \in \mathcal{S}(\mathsf{C})} D_\alpha^*\left( \omega_{A'C} \Vert \omega_{A'} \otimes \sigma_C \right).
\end{align}
The term `{\color{BrickRed}\texttt{mother}}' stands for the mother protocol \cite{Horodecki2005}, while `{\color{RoyalBlue}\texttt{father}}' stands for the father protocol \cite{Proc465} or the quantum splitting protocol. The $E^*_\textnormal{\texttt{mother}}(R)_{\rho_{AE}}
 $ is motivated from Dupuis' result \cite{Dup23} that, for $\mathcal{T}_{A\to C} = \Tr_{A\backslash C}$ being a partial trace,
\begin{align} \label{eq:mother}
	\frac12\mathbf{E}_{\mathds{U}(\mathsf{A})} \left\| \Tr_{A\backslash C}\left(U_A \rho_{AE} U_A^\dagger \right) - \omega_{C}\otimes\rho_E \right\|_1 \leq \e^{ -
	{\color{BrickRed}
	E^*_\textnormal{\texttt{mother}}(\log |C|)_{\rho_{AE}}
	}	
},
\end{align}
where
\begin{align}
	E^*_\textnormal{\texttt{mother}}(\log |C|)_{\rho_{AE}} > 0
	\quad \Longleftrightarrow \quad
	\log |C| < \frac12 \left( H\left( A{\,|\,} E \right)_\rho + \log |A| \right).
\end{align}
The bound in \eqref{eq:mother} provides an exponential achievability for the standard decoupling (i.e.~$\mathcal{T}_{A\to C} = \Tr_{A\backslash C}$), which in turns provides the maximum (one-shot) distillable entanglement ($\log |C|$) in the mother protocol.
Also note that \eqref{eq:mother} can directly be derived from Theorem~\ref{theo:exp} (up to a constant) by noting that $H_\alpha^*(A'{\,|\,}C)_\omega  = \log \frac{|A|}{|C|^2}$.


Then, the other error exponent that competes with ${\color{BrickRed}
	E^*_\textnormal{\texttt{mother}}(R)_{\rho_{AE}}
}  $ should correspond to the father protocol i.e.~${\color{RoyalBlue}
E^*_\textnormal{\texttt{father}}(R)_{\omega_{A' C}}
}$.
Noting the Choi-Jamio{\l}kowski state $\omega_{A'C}$ in Theorem~\ref{theo:exp}, we have
\begin{align}
	H_\alpha^*(A'{\,|\,}C)_{\omega} - \log |A| = - I_\alpha^*(A'{\,:\,}C)_\omega
\end{align}
because $\omega_{A'} = \frac{ \mathds{1}_{A'} }{ |A'| }$ and $|A'| = |A|$ as the Choi state of a quantum channel. Hence, we obtain
\begin{align}
	&\frac{\alpha-1}{\alpha} \left( H_\alpha^*(A{\,|\,}E)_\rho + H_\alpha^*(A'{\,|\,}C)_{\omega} \right)\\
	&= \frac{\alpha-1}{\alpha} \left( H_\alpha^*(A{\,|\,}E)_\rho + \log |A| \right) + \frac{\alpha-1}{\alpha} \left( - I_\alpha^*(A'{\,:\,}C)_\omega \right) \\
	&= {\color{BrickRed}\frac{\alpha-1}{\alpha} \left( H_\alpha^*(A{\,|\,}E)_\rho + \log |A| - \log |C|^2 \right)}  +
	{\color{RoyalBlue}\frac{\alpha-1}{\alpha} \left( \log |C|^2 - I_\alpha^*(A'{\,:\,}C)_\omega \right) }.
\end{align}
This suggests that the above expression may be interpreted as the sum of two error exponent functions (by choosing the worst rate).
More precisely, by using the concavity in Lemma~\ref{lemm:concave}, we obtain the following expression of the error exponent function via Fenchel's duality.
It would be very interesting to find out what protocol the exponent ${\color{RoyalBlue}
	E^*_\textnormal{\texttt{father}}(\log |C|)_{\omega_{A' C}}
}$ corresponds to.
\begin{prop}[A Fenchel duality for exponent functions] \label{prop:Fenchel}
	Following the above notations, we have
	\begin{align}
	\sup_{\alpha\in[1,2]}	\frac{\alpha-1}{\alpha} \left( H_\alpha^*(A{\,|\,}E)_\rho + H_\alpha^*(A'{\,|\,}C)_{\omega} \right)
	=
	\inf_{R\geq 0} \left\{
	{\color{BrickRed}
		E^*_\textnormal{\texttt{mother}}(\log |C|)_{\rho_{AE}}
	}
	+
	{\color{RoyalBlue}
		E^*_\textnormal{\texttt{father}}(\log |C|)_{\omega_{A' C}}
	}
	\right\}.
	\end{align}
\end{prop}
\begin{proof}
	We first recall Fenchel's duality theorem \cite{Roc70}.
	Let $f$ (resp.~$g$) be proper convex function (resp.~proper concave function) from some Banach space to extended real lines.
	Then,
	\begin{align}
		\inf_x \left\{ f(x) - g(x)  \right\} = \sup_{x^*}\left\{ g_*(x^*) - f^*(x^*)   \right\},
	\end{align}
	where
	\begin{align}
		f^*(x^*) := \sup_{x} \left\{ \langle x^*, x \rangle - f(x)  \right\}
	\end{align}
	is the convex conjugate of $f$, $\langle \cdot, \cdot \rangle$ denotes an inner product, and  similarly for $g_*$.

	In view of the definitions given above, we let
	\begin{align}
		\begin{cases}
			x = \frac{\alpha-1}{\alpha} \in [-\sfrac{1}{2}, 0] \\
			-f(x) = \frac{\alpha-1}{\alpha} \left( H_\alpha^*(A{\,|\,}E)_\rho + \log |A| \right)\\
			g(x) = \frac{\alpha-1}{\alpha} \left(H_\alpha^*(A'{\,|\,}C)_{\omega} - \log |A| \right) = \frac{1-\alpha}{\alpha} I_\alpha^*(A'{\,:\,}C)_\omega \\
			x^* = R
		\end{cases}
	\end{align}
	The the claim follows from the concavity of $x\mapsto - f(x)$ and $x\mapsto g(x)$, which is delayed to Lemma~\ref{lemm:concave} by using similar idea in \cite{CGH18}.
\end{proof}

With Proposition~\ref{prop:Fenchel}, we can interpret the error exponent in Theorem~\ref{theo:exp} as in Figure~\ref{fig:exponent}.
\begin{figure}[h!]
	\centering
	\includegraphics[width = 0.8\linewidth]{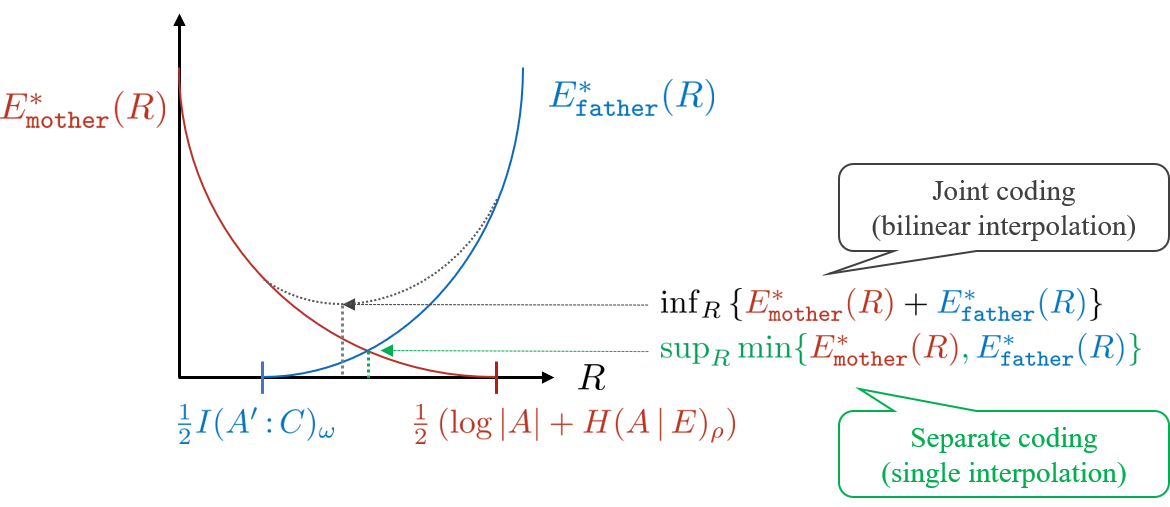}
	\caption {Interpretation of the error exponent.}
	\label{fig:exponent}
\end{figure}
The error exponent in Theorem~\ref{theo:exp} is positive if and only if
\begin{align}
	\quad H(A{\,|\,}E)_\rho + H(A'{\,|\,}C)_{\omega} \geq 0
	\Longleftrightarrow
	\frac12 \left( H\left( A{\,|\,} E \right)_\rho + \log |A| \right) > \frac12 I_\alpha^*(A'{\,:\,}C)_\tau.
\end{align}
This can be viewed as the quantum version of the \emph{joint source-channel coding with side information} in the classical Shannon theory. 

\section{Acknowledgement} LG was partially supported by NSF grant DMS-2154903. HC is supported by the National Science and Technology Council, Taiwan (R.O.C.) under Grants No.~NSTC 112-2636-E-002-009, No.~NSTC 113-2119-M-007-006, No.~NSTC 113-2119-M-001-006, No.~NSTC 113-2124-M-002-003, and No.~NSTC 113-2628-E-002-029 by the Yushan Young Scholar Program of the Ministry of Education, Taiwan (R.O.C.) under Grants No.~NTU-112V1904-4 and by the research project ``Pioneering Research in Forefront Quantum Computing, Learning and Engineering'' of National Taiwan University under Grant No. NTU-CC- 112L893405 and NTU-CC-113L891605. H.-C.~Cheng acknowledges the support from the “Center for Advanced Computing and Imaging in Biomedicine (NTU-113L900702)” through The Featured Areas Research Center Program within the framework of the Higher Education Sprout Project by the Ministry of Education (MOE) in Taiwan.

\appendix
\section{Miscellaneous Lemmas} \label{sec:lemmas}

\begin{lemm}[Concavity of scaled sandwiched R\'enyi conditional entropy] \label{lemm:concave}
	For any bipartite density operator $\rho_{AB}$, the following map
	\begin{align}
		s \mapsto -s H_{\frac{1}{1+s}}^*(A{\,|\,}B)_\rho
	\end{align}
	is concave on $s \in [0,-1]$.
\end{lemm}
\begin{proof}This follows from the interpolation relation of $S_1(\mathsf{B}, S_\frac{1}{s+1}(\mathsf{A}))$ norm proved by Pisier \cite{pisier1998non}. By the definition of $H_{\frac{1}{1+s}}^*(A{\,|\,}B)_\rho$,
\[f(s)=-s H_{\frac{1}{1+s}}^*(A{\,|\,}B)_\rho=-\log \norm{ \rho}{S_1(\mathsf{B}, S_\frac{1}{s+1}(\mathsf{A}))} \pl.\]
Take the convex combination $s=(1-\theta)s_0+\theta s_1$ of $s_0,s_1\in [0,-1], \theta[0,1] $. This implies that for $p=\frac{1}{1+s}$ and $p_i=\frac{1}{1+s_i}, i=1,2$,
\[ \frac{1}{p}=\frac{1-\theta}{p_0}+\frac{\theta}{p_1}\pl\]
By the interpolation inequality \eqref{eq:interpolation} in Appendix \ref{sec:interpolation},
\begin{align*}
f(s)=&-\log \norm{ \rho}{S_1(\mathsf{B}, S_{\frac{1}{s+1}}(\mathsf{A}))}
\\ =&-\log \norm{ \rho}{S_1(\mathsf{B}, S_p(\mathsf{A}))}
\\ \ge &-\log \norm{ \rho}{S_1(\mathsf{B}, S_{p_0}(\mathsf{A}))}^{1-\theta} \cdot \norm{ \rho}{S_1(\mathsf{B}, S_{p_1}(\mathsf{A}))}^{\theta}
\\ =& -(1-\theta)\log \norm{ \rho}{S_1(\mathsf{B}, S_{p_0}(\mathsf{A}))}- \theta\log \norm{ \rho}{S_1(\mathsf{B}, S_{p_1}(\mathsf{A}))}
\\ =& (1-\theta)f(s_0)+\theta f(s_1)\qedhere
\end{align*}
\end{proof}

\begin{lemm}[\hspace{1sp}{\cite[Lemmas 3.4 \& 3.5]{DBW+14}}] \label{lemm:swap}
	Let $F_A =\sum_{i,j} |i\rangle \langle j| \ten |j\rangle \langle i| \in \mathcal{B}(\mathsf{A}\otimes \mathsf{\tilde{A}})$ be the swap operator between $\mathsf{A}$ and its copy $\mathsf{\tilde{A}}$. Then
	for $M_A \in \mathcal{B}(\mathsf{A})$ and $N_{\tilde{A}} \in  \mathcal{B}(\mathsf{\tilde{A}}) \cong \mathcal{B}(\mathsf{A})$,
	\[
	\tr[M_A N_A]=\tr[(M_A\ten N_{\tilde{A}})F_A].
	\]
	For an operator $ Y_{A\tilde{A}} \in  \mathcal{B(\mathsf{A}\otimes \mathsf{\tilde{A}'})}$,
	\[
	\mathbf{E}_{ \mathds{U}(\mathsf{A}) } (U_A\ten U_{\tilde{A}}) Y_{A\tilde{A}} (U_A^\dagger\ten U_{\tilde{A}}^\dagger )
	=\al \mathds{1}_{A \tilde{A}} +\beta F_A
	\]
	with the scalars $\alpha$ and $\beta$ satisfying the following:
	\[
	\al d^2+\beta d=\tr[Y_{A\tilde{A}}]\pl, \quad
	\al d+\beta d^2=\tr[Y_{A\tilde{A}} F_A]\pl.
	\]
	where $d=|A|$.
\end{lemm}

\begin{lemm} \label{lemm:phi_twirling}
	Let $A$ be a Hilbert space with dimension $|A|=d$ and $\mathsf{A'},\mathsf{\tilde{A}},\mathsf{\tilde{A}'}$ be copies of $A$, Then
	\begin{align}
	\mathbf{E}_{\mathds{U}(\mathsf{A})}  \left[\tilde{\Phi}_{A'A}^{U^\dagger}\ten \tilde{\Phi}_{\tilde{A}\tilde{A}'}^{U^\dagger} \right]
	=\left( \frac{1}{d^2-1} \mathds{1}_{AA'\tilde{A}\tilde{A}'}-\frac{1}{d^3-d}F_A\ten \mathds{1}_{A'\tilde{A}'}-\frac{1}{d^3-d} F_{A'}\ten \mathds{1}_{A\tilde{A}}+\frac{1}{d^2-1}F_{AA'}\right),
	\end{align}
	where $F_A =\sum_{i,j} |i\rangle \langle j| \ten |j\rangle \langle i| \in \mathcal{B}(\mathsf{A}\otimes \mathsf{\tilde{A}})$ be the swap operator between $\mathsf{A}$ and its copy $\mathsf{\tilde{A}}$, $\tilde{\Phi}_{AA'} = \sum_{i,j} |i\rangle\langle j| \otimes |i\rangle \langle j|$ be the non-normalized maximally entangled state between $\mathsf{A}$ and $\mathsf{A}'$ and $\tilde{\Phi}_{AA'}^U=(U_A\ten \mathds{1}_{A'})\Phi_{AA'}(U_A\ten \mathds{1}_{A'})$ be the rotated version by $U$ acting on $A$ (resp. $\tilde{A}$) .
\end{lemm}

\begin{proof}
 Note that $\tilde{\Phi}_{A'A}\ten \tilde{\Phi}_{\tilde{A}\tilde{A}'}=\Big(\sum_{ij}e_{ij}\ten e_{ij}\Big)\ten \Big(\sum_{kl}e_{kl}\ten e_{kl}\Big)$ by denoting $e_{ij} := |i\rangle \langle j| $ for simplicity. Using the fact
\[ \tr[e_{ij}\ten e_{kl}]=
\begin{dcases}
	1, & \mbox{if } i=j,k=l \\
	0, & \mbox{otherwise}.
\end{dcases}\pl \text{ and } \pl
\tr[e_{ij}e_{kl}]
=\begin{dcases}
	1, & \mbox{if } j=k,i=l \\
	0, & \mbox{otherwise}.
\end{dcases}\]
we have by Lemma \ref{lemm:phi_twirling}, on $A\ten \tilde{A}$
\begin{align*}
	\int_{\mathds{U}(\mathsf{A})} U\ten U (e_{ij}\ten e_{kl})U^\dagger\ten U^\dagger \, \mathrm{d}U=
	\begin{dcases}
		\frac{1}{d(d+1)} \mathds{1}_{A\tilde{A}}+\frac{1}{d(d+1)}F_A, & \mbox{if } i=j=k=l \\
		\frac{1}{d^2-1}\mathds{1}_{A\tilde{A}}-\frac{1}{d^3-d}F_A, & \mbox{if } i=j\neq k=l \\
		-\frac{1}{d^3-d}\mathds{1}_{A\tilde{A}} +\frac{1}{d^2-1}F_A, & \mbox{if } i=l\neq j=k \\
		0, & \mbox{otherwise}.
	\end{dcases}
\end{align*}
Note that $\frac{1}{d(d+1)}=\frac{1}{d^2-1}-\frac{1}{d^3-d}$. Then,
\begin{align*}
	&\mathbf{E}_{\mathds{U}(\mathsf{A})}  \left[\tilde{\Phi}_{A'A}^{U^\dagger}\ten \tilde{\Phi}_{\tilde{A}\tilde{A}'}^{U^\dagger} \right] \\
	&= \sum_{i} e_{ii}\ten \left(\frac{1}{d(d+1)} \mathds{1}_{A\tilde{A}} +\frac{1}{d(d+1)}F_A \right)\ten e_{ii}+\sum_{i\neq k}e_{ii}\ten \left(\frac{1}{d^2-1} \mathds{1}_{A\tilde{A}} -\frac{1}{d^3-d} F_A \right)\ten e_{kk}\\
	&\quad +\sum_{i\neq j}e_{ij}\ten \left(-\frac{1}{d^3-d} \mathds{1}_{A\tilde{A}} +\frac{1}{d^2-1} F_A\right)\ten e_{ji}\\
	&=\sum_{i, k} e_{ii}\ten \left(\frac{1}{d^2-1}\mathds{1}_{A\tilde{A}}-\frac{1}{d^3-d}F_A\right)\ten e_{kk}+\sum_{i, j}e_{ij}\ten \left(-\frac{1}{d^3-d}\mathds{1}_{A\tilde{A}}+\frac{1}{d^2-1}F_A\right)\ten e_{ji}\\
	&=\mathds{1}_{A'\tilde{A}'} \ten \left(\frac{1}{d^2-1}\mathds{1}_{A\tilde{A}'}-\frac{1}{d^3-d}F_A \right)+F_{A'}\ten \left(-\frac{1}{d^3-d}\mathds{1}_{A\tilde{A}}+\frac{1}{d^2-1}F_A \right)\\
	&=\frac{1}{d^2-1}\mathds{1}_{AA'\tilde{A}\tilde{A}'}-\frac{1}{d^3-d}F_A\ten\mathds{1}_{A
'\tilde{A}'} -\frac{1}{d^3-d}F_{A'}\ten \mathds{1}_{A\tilde{A}}+\frac{1}{d^2-1}F_{AA'}.
\end{align*}
\end{proof}

\section{Complex Interpolation} \label{sec:interpolation}

In this section, we briefly review the definition of the complex interpolation. We refer to \cite{bergh2012interpolation} for a detailed account of interpolation spaces.  Let $X_0$ and $X_1$ be two Banach spaces. Assume that there exists a Hausdorff topological vector space $X$ such that $X_0, X_1\subset X$ as subspaces. Let $\mathcal{S}=\{z:0\le \texttt{Re} (z)\le 1\}$ be the unit vertical strip on the complex plane, and $\mathcal{S}_0=\{z : 0< \texttt{Re} (z)< 1\}$ be its open interior. Let $\F(X_0, X_1)$ be the space of all functions $f:\mathcal{S}\to X_0+X_1$, which are bounded and continuous on $\mathcal{S}$ and analytic on $\mathcal{S}_0$, and moreover
\[\{f(\mathrm{i}t)\pl : t\in \mathbb{R}\}\subset X_0\pl ,\pl \{f(1+\mathrm{i}t)\pl : t\in \mathbb{R}\}\subset X_1\pl.\]
$\F(X_0, X_1)$ is again a Banach space equipped with the norm
\[\norm{f}{\F} :=\max\left\{\pl \sup_{t\in \mathbb{R}} \norm{f(\mathrm{i}t)}{X_0}\pl,\pl \sup_{t\in \mathbb{R}}\norm{f(1+\mathrm{i}t)}{X_1}\right\}\pl. \]
The complex interpolation space $(X_0,X_1)_\theta$, for $0\le \theta\le 1$, is the quotient space of $\F(X_0,X_1)$ as follows,
\[(X_0, X_1)_\theta=\{\pl x\in X_0+X_1 :  x=f(\theta) \pl \text{for some }\pl  f\in \F(X_0, X_1)\pl\} \pl.\]
where quotient norm is
\begin{align}\label{def}\norm{x}{\theta}=\inf \{\pl\norm{f}{\F} \pl : f(\theta)=x \pl\}\pl .\end{align}
It is clear from the definition that $X_0=(X_0,X_1)_0,  X_1=(X_0,X_1)_1$. For all $0<\theta<1$,
$(X_0,X_1)_\theta$ are called interpolation space of $(X_0,X_1)$.

The most basic example is the $p$-integrable function spaces $L_p(\Omega,\mu)$ of a positive measure space $(\Omega,\mu)$.
For $1\le p\le \infty$, $L_p(\Omega,\mu)$  forms a family of interpolation spaces, i.e.
\begin{align}\label{eq:in} L_p(\Omega,\mu)\cong [L_{p_0}(\Omega,\mu),{L_{p_1}}(\Omega,\mu)]_\theta \end{align}
holds isometrically for all $1\le p_0,p_1,p\le \infty, 0\le \theta\le 1$ such that
$\frac{1}{p}=\frac{1-\theta}{p_0}+\frac{\theta}{p_1}$. For a von Neumann algebra $(\mathcal{M},\text{Tr})$ equipped with normal faithful semifinite trace $\text{Tr}$,
the noncommutative $L_p$-norm is defined as $\|x\|_p=\text{Tr}[|x|^p]^{\sfrac{1}{p}}$ and $L_p(\mathcal{M},\text{Tr})$ (or shortly $L_p(\mathcal{M})$) is the completion of $\{x\in \mathcal{M} : \|x\|_p<\infty\}$. The noncommutative analog of \eqref{eq:in} is that
\begin{align}\label{eq:in2} L_p(\mathcal{M},\text{Tr})\cong [L_{p_0}(\mathcal{M},\text{Tr}),{L_{p_1}}(\mathcal{M},\text{Tr})]_\theta .\end{align}
In particular, the Schatten $p$-class on a Hilbert space $\mathcal{H}$ are the $L_p$ spaces of $(\mathcal{B(H)},\tr)$ which satisfies
\[S_p(\mathcal{H}) \cong \left[S_{p_0}(\mathcal{H}), S_{p_1}(\mathcal{H}) \right]_{\theta}\pl.\]
Here $S_\infty(\mathcal{H})$ is identified with $\mathcal{B(H)}$.

The complex interpolation relation has been already used in many works in quantum information theory, e.g.~\cite{CGH18}. In this paper, we use the complex interpolation for vector valued $L_p$ space. The first one is a classical-quantum mixture of \eqref{eq:in} and \eqref{eq:in2}.
   For an operator-valued function $f:\Omega \to \mathcal{B(H)}$, its $L_p$ norm is given by
\begin{align} \label{eq:Lp} \norm{f}{L_p(\Omega, S_p(\mathcal{H}))}:=\left(\int_{\Omega} \norm{f(\omega)}{S_p(\mathcal{H})}^p \d\mu(\omega)\right)^{\sfrac{1}{p}}.
\end{align}
Note that $L_p(\Omega, S_p(\mathcal{H}))$ is exactly the $L_p$-space of classical-quantum system $L_\infty(\Omega,\mathcal{B(H)})$, which is a von Neumann algebra equipped with the trace $\tau(f)=\int_{\Omega}\tr[f(\omega)] \,\mathrm{d}\mu(\omega)$. Thus $L_p(\Omega, S_p(\mathcal{H}))$ satisfies complex interpolation by \eqref{eq:in2}. That is, for $\frac{1}{p}=\frac{1-\theta}{p_0}+\frac{\theta}{p_1}$ and $p_0,p_1\in [1,\infty]$.
\[ L_{p}(\Omega,S_{p}(\mathcal{H}))\cong [L_{p_0}(\Omega,S_{p_0}(\mathcal{H})),L_{p_1}(\Omega,S_{p_1}(\mathcal{H}))]_\theta\]
In particular,
$L_2(\Omega, S_2(\mathcal{H}))$ is a Hilbert space with inner product.
\[ \langle f, g\rangle = \int_{\Omega}\tr[f(\omega)^*g(\omega)]\,\d\mu(\omega).\]
Another more advanced interpolation spaces we use is Pisier's noncommutative vector-valued $L_p$-space introduced in \cite{pisier1998non} (see Section \ref{sec:preliminaries} for definitions). That is, for $\frac{1}{p}=\frac{1-\theta}{p_0}+\frac{\theta}{p_1}, \frac{1}{q}=\frac{1-\theta}{q_0}+\frac{\theta}{q_1}$, and $\theta\in[0,1], p_0,p_1,q_0,q_1\in [1,\infty]$,
\begin{align} \label{eq:complex} S_q(\mathsf{B},S_{p}(\mathsf{A}))=[S_{q_0}(\mathsf{B},S_{p_0}(\mathsf{A})),S_{q_1}(\mathsf{B},S_{p_1}(\mathsf{A}))]_{\theta}\ ,
\end{align}

For these interpolation spaces, we will use the following Riesz–Thorin interpolation theorem.

\begin{theo}[Riesz--Thorin interpolation theorem] \label{thm:interpolation} Let $(X_0, X_1)$ and $(Y_0, Y_1)$ be two compatible couples of Banach spaces and let $(X_0, X_1)_\theta$ and $(Y_0, Y_1)_\theta$ be the corresponding interpolation space of exponent $\theta$. Suppose $T:X_0 + X_1 \to Y_0 + Y_1$, is a linear operator bounded from $X_j$ to $Y_j$, $j=0, 1$. Then $T$ is bounded from $(X_0, X_1)_\theta$ to $(Y_0, Y_1)_\theta$, and moreover,
\[ \|T:(X_0, X_1)_\theta\to (Y_0, Y_1)_\theta\|_{}\leq \|T:X_0\to Y_0\|_{}^{1-\theta }\|T:X_1\to Y_1\|_{}^{\theta}.\]
\end{theo}

In particular, for $S_1(\mathsf{B},S_{p}(\mathsf{A}))$, the complex interpolation \eqref{eq:complex}  implies the following interpolation inequality:
\begin{align}\label{eq:interpolation}
	\left\|X \right\|_{ S_1(\mathsf{B}, \, S_{p_\theta}(\mathsf{A})) }
	\leq \left\|X \right\|_{ S_1(\mathsf{B}, \, S_{p_0}(\mathsf{A})) }^{1-\theta}
	\left\|X \right\|_{ S_1(\mathsf{B}, \,S_{p_1}(\mathsf{A})) }^{\theta}.
\end{align}

\section{Multiplicativity of $S_1(S_\al)$-norm}\label{sec:multi}
\begin{prop}
 \label{lemm:multiplicative}
 Let $\alpha \in (0,1) \cup (1,\infty]$ be arbitrary.
For bounded operators $\omega\in \mathcal{B}(\mathsf{A'}\otimes \mathsf{C})$ and $\rho\in \mathcal{B}( \mathsf{A} \otimes \mathsf{E})$,
\begin{align}
	\norm{\omega\ten \rho}{S_1(\mathsf{C}\otimes \mathsf{E},S_\al(\mathsf{A}\otimes \mathsf{A'}))}=\norm{\omega}{S_1(\mathsf{C},S_\al(\mathsf{A'}))}\norm{\rho}{S_1(\mathsf{E},S_\al(\mathsf{A}))}.
\end{align}
In particular, for density operators $\omega\in \mathcal{S}(\mathsf{A'}\otimes \mathsf{C})$ and $\rho\in \mathcal{S}( \mathsf{A} \otimes \mathsf{E})$,
\begin{align} \label{eq:additivity_entropy}
	H_\al^*(A' A {\,|\,} C E)_{\omega \otimes \rho}
	=
	H_\al^*(A' {\,|\,} C)_\omega + H_\al^*(A {\,|\,} E)_\rho.
\end{align}

\end{prop}

\begin{proof}
By definition, one direction is straightforward, i.e.~
\begin{align*} &\norm{\omega\ten \rho}{S_1(\mathsf{C}\otimes \mathsf{E},S_\al(\mathsf{A}\otimes \mathsf{A'}))}\\
= &\inf_{\sigma_{CE},\gamma_{CE}}\norm{ \sigma_{CE}^{-\frac{1}{2\al'}}(\omega\ten \rho)\gamma_{CE}^{-\frac{1}{2\al'}}}{\al}\\
\le &\inf_{\sigma_{C}\ten \sigma_{E},\gamma_{C}\ten \gamma_{E}}\norm{ (\sigma_{C}\ten \sigma_{E})^{-\frac{1}{2\al'}}(\omega\ten \rho)(\gamma_{C}\ten \gamma_{E})^{-\frac{1}{2\al'}}}{\al}\\
=&\inf_{\sigma_{C},\gamma_{C}}\norm{ \sigma_{C}^{-\frac{1}{2\al'}}\omega_{A'C}\gamma_{C}^{-\frac{1}{2\al'}}}{\al}
\inf_{\sigma_{E},\gamma_{E}}\norm{ \sigma_{E}^{-\frac{1}{2\al'}}\rho_{AE}\gamma_{E}^{-\frac{1}{2\al'}}}{\al}\\
=&\norm{\omega}{S_1(\mathsf{C},S_\al(\mathsf{A'}))}\norm{\rho}{S_1(\mathsf{E},S_\al(\mathsf{A}))},
\end{align*}
where the infimums are over all density operators respectively $\sigma_{C},\sigma_E$, $\sigma_{CE}$ and similarly for $\gamma$.
The converse direction needs several steps.

Step 1. We show  that for positive $X_{A'C}$ and $Y_{AE}$,
\begin{align*}
	\norm{X_{A'C}\ten Y_{AE}}{S_\infty( \mathsf{C}\otimes \mathsf{E},  S_1( \mathsf{A}\otimes \mathsf{A'}) )}=\norm{X_{A'C}}{S_\infty(\mathsf{C},  S_1(\mathsf{A'} ) )}\norm{Y_{AE}}{S_\infty( \mathsf{E},  S_1(\mathsf{A} ) )}
\end{align*}
Indeed, for a positive $Y_{AE}$,
\begin{align*}\norm{Y_{AE}}{S_\infty(  \mathsf{E},  S_1( \mathsf{A}) )}
&=\sup_{\sigma_{E},\gamma_E: \|\sigma_{E}\|_1=\|\gamma_{E}\|_1\leq 1  } \left\|\sigma^{1/2}_{E} Y_{AE} \gamma^{1/2}_{E}\right\|_1\\
&\le \left(\sup_{\sigma_{E}: \|\sigma_{E}\|_1\leq 1  } \left\|\sigma^{1/2}_{E} Y_{AE}^{1/2} \right\|_2\right)\cdot  \left(\sup_{\gamma_{E}: \|\gamma_{E}\|_1 \leq 1  } \left\|Y_{AE}^{1/2} \gamma^{1/2}_{E} \right\|_2\right)\\
&=\left(\sup_{\sigma_{E}: \|\sigma_{E}\|_1\leq 1  } \left\|\sigma^{1/2}_{E} Y_{AE}^{1/2} \right\|_2\right)^2
\\
&=\sup_{\sigma_{E}: \|\sigma_{E}\|_1\leq 1  } \left\|\sigma^{1/2}_{E} Y_{AE} \sigma^{1/2}_{E}\right\|_1
\\
&=\sup_{\sigma_{E}: \|\sigma_{E}\|_1\leq 1  } \tr_{AE}[\sigma_{E} Y_{AE}]
\\
&=\sup_{\sigma_{E}: \|\sigma_{E}\|_1\leq 1  } \tr_{E}[\sigma_{E} Y_{E}]
\\
&=\norm{Y_E}{\infty},
\end{align*}
where we apply the duality between the Schatten $1$-norm and operator norm.

Applying this property to $X_{A'C}\ten Y_{AE}$, we have
\begin{align*}\norm{X_{A'C}\ten Y_{AE}}{S_\infty( \mathsf{C}\otimes \mathsf{E},  S_1( \mathsf{A}\otimes \mathsf{A'}) )}
&=\norm{X_{C}\ten Y_E}{\infty} \\
&=\norm{X_C}{\infty}\norm{Y_E}{\infty} \\
&=\norm{X_{A'C}}{S_\infty(\mathsf{C},  S_1(\mathsf{A'} ) )}\cdot\norm{Y_{AE}}{S_\infty( \mathsf{E},  S_1(\mathsf{A} ) )}.
\end{align*}

Step 2. We show  that for  general $X_{A'C}$ and $Y_{AE}$, i.e.~
\begin{align*}\norm{X_{A'C}\ten Y_{AE}}{S_\infty( \mathsf{C}\otimes \mathsf{E},  S_1( \mathsf{A}\otimes \mathsf{A'}) )}=\norm{X_{A'C}}{S_\infty(\mathsf{C},  S_1(\mathsf{A'} ) )}\norm{Y_{AE}}{S_\infty( \mathsf{E},  S_1(\mathsf{A} ) )}.
\end{align*}
We need the following factorization formula
\begin{align}
	\norm{X_{A'C}}{S_\infty( \mathsf{C},S_1(\mathsf{A'}))}
=\sup_{X=ab} \norm{ aa^*}{S_\infty( \mathsf{C},S_1(\mathsf{A'}))}^{1/2}\norm{b^*b}{S_\infty( \mathsf{C},S_1(\mathsf{A'}))}^{1/2}. \label{eq:factorization}
\end{align}
One direction is easy, that for any factorization $X=ab$
\begin{align*}
	\norm{X_{A'C}}{S_\infty( \mathsf{C},S_1(\mathsf{A'}))}
	&=\sup_{\sigma,\gamma}\norm{\sigma_C^{1/2} X\gamma_C^{1/2}}{1}\\
	&\le \sup_{\sigma,\gamma}\norm{\sigma_C^{1/2} a}{2}\norm{b\gamma_C^{1/2}}{2}\\
&= \sup_{\sigma,\gamma}\norm{\sigma_C^{1/2} aa^*\sigma_C^{1/2}}{1}^{1/2}\norm{\gamma_C^{1/2}b^*b\gamma_C^{1/2}}{1}^{1/2}\\
&= \norm{ aa^*}{S_\infty( \mathsf{C},S_1(\mathsf{A'}))}^{1/2}\norm{b^*b}{S_\infty( \mathsf{C},S_1(\mathsf{A'}))}^{1/2}.
\end{align*}
To see the inequality is attained,
we first find optimal $\sigma$ and $\gamma$ in the definition of $\norm{X}{S_\infty( \mathsf{C},S_1(\mathsf{A'}))}$. Then \[\norm{X}{S_\infty( \mathsf{C},S_1(\mathsf{A'}))}=\norm{\sigma^{1/2} X\gamma^{1/2}}{1}=\norm{X_1}{2}\norm{X_2}{2}\] for some $X_1,X_2$ with $X_1X_2=\sigma^{1/2} X\gamma^{1/2}$. Then $a=\sigma^{-1/2}X_1, b=X_2\gamma^{-1/2}$ is the optimal factorization. Now choose $X=ab$ and $Y=cd$ with the above equality \eqref{eq:factorization} being attained. Then, for $X\ten Y=(a\ten c)(b\ten d)$,
\begin{align*}
	\norm{X_{A'C}\ten Y_{AE}}{S_\infty(\mathsf{C}\otimes \mathsf{E},S_1(\mathsf{A}\otimes\mathsf{A'}))}&\le \norm{aa^*\ten cc^*}{S_\infty(\mathsf{C}\otimes \mathsf{E},S_1(\mathsf{A}\otimes\mathsf{A'}))}^{1/2}\norm{b^*b\ten d^*d}{S_\infty(\mathsf{C}\otimes \mathsf{E},S_1(\mathsf{A}\otimes\mathsf{A'}))}^{1/2}\\
&\le \norm{aa^*}{S_\infty( \mathsf{C},S_1(\mathsf{A'}))}^{1/2}\norm{ cc^*}{S_\infty(E,S_1(A))}^{1/2}\norm{ d^*d}{S_\infty(E,S_1(A))}^{1/2}\norm{b^*b}{S_\infty( \mathsf{C},S_1(\mathsf{A'}))}^{1/2}\\
&=\norm{X_{A'C}}{S_\infty(\mathsf{C},  S_1(\mathsf{A'} ) )}\norm{Y_{AE}}{S_\infty( \mathsf{E},  S_1(\mathsf{A} ) )}.
\end{align*}
where in the second inequality we used Step 1. The other direction is easy by definition
\begin{align*}\norm{X_{A'C}\ten Y_{AE}}{S_\infty( \mathsf{C}\otimes \mathsf{E},  S_1( \mathsf{A}\otimes \mathsf{A'}) )}\ge \norm{X_{A'C}}{S_\infty(\mathsf{C},  S_1(\mathsf{A'} ) )}\norm{Y_{AE}}{S_\infty( \mathsf{E},  S_1(\mathsf{A} ) )}.
\end{align*}

Step 3. We show
\[ \norm{X_{A'C}\ten Y_{AE}}{S_\infty( \mathsf{C}\otimes \mathsf{E},S_p(\mathsf{A}\otimes \mathsf{A'}))}
=\norm{X_{A'C}}{S_\infty(\mathsf{C},  S_p(\mathsf{A' }) )}\norm{Y_{AE}}{S_\infty(\mathsf{E},  S_p(\mathsf{A} ) )}\]
for any $1\le p\le \infty$ by interpolation. Consider the map \begin{align*}
	&{t}: S_\infty(\mathsf{C},  S_p(\mathsf{A' }) )\times S_\infty(\mathsf{E},  S_p(\mathsf{A} ) )\to S_\infty( \mathsf{C}\otimes \mathsf{E},S_p(\mathsf{A}\otimes \mathsf{A'}))\\
& t(X,Y) :=X\ten Y
\end{align*}
We claim that this bilinear map is a contraction. The case of $p=1$ was proved above. As for, $p=\infty$, noting that $S_\infty( \mathsf{E},  S_\infty( \mathsf{A} ))=\mathcal{B}(\mathsf{E}\ten \mathsf{A})$, the claim  is obvious.
Then for general $p \in (1,\infty)$,
using complex interpolation
\begin{align}
	S_\infty(\mathsf{C},  S_p(\mathsf{A' }) ) &= \left[ S_\infty(\mathsf{C},  S_1(\mathsf{A' }) ), S_\infty(\mathsf{C},  S_\infty(\mathsf{A' }) ) \right]_\theta; \\
	S_\infty(\mathsf{E},  S_p(\mathsf{A} ) ) &= \left[ S_\infty(\mathsf{E},  S_1(\mathsf{A} ) ), S_\infty(\mathsf{E},  S_1(\mathsf{A} ) ) \right]_\theta; \\
	S_\infty( \mathsf{C}\otimes \mathsf{E},S_p(\mathsf{A}\otimes \mathsf{A'})) &= \left[ S_\infty( \mathsf{C}\otimes \mathsf{E},S_1(\mathsf{A}\otimes \mathsf{A'})), S_\infty( \mathsf{C}\otimes \mathsf{E},S_\infty(\mathsf{A}\otimes \mathsf{A'})) \right]_\theta,
\end{align}
for $\theta = \frac{p-1}{p} \in (0,1)$,
we have
\begin{align}
	&\left\|{t}: S_\infty(\mathsf{C},  S_p(\mathsf{A' }) )\times S_\infty(\mathsf{E},  S_p(\mathsf{A} ) )\to S_\infty( \mathsf{C}\otimes \mathsf{E},S_p(\mathsf{A}\otimes \mathsf{A'})) \right\| \\
	&\leq \left\|{t}: S_\infty(\mathsf{C},  S_1(\mathsf{A' }) )\times S_\infty(\mathsf{E},  S_1(\mathsf{A} ) )\to S_\infty( \mathsf{C}\otimes \mathsf{E},S_1(\mathsf{A}\otimes \mathsf{A'})) \right\|^{1-\theta}\\
	&\quad \cdot\left\|{t}: S_\infty(\mathsf{C},  S_\infty(\mathsf{A' }) )\ten S_\infty(\mathsf{E},  S_\infty(\mathsf{A} ) )\to S_\infty( \mathsf{C}\times \mathsf{E},S_\infty(\mathsf{A}\otimes \mathsf{A'})) \right\|^\theta \\
	&=1^{1-\theta} \cdot 1^{\theta} = 1,
\end{align}
which implies that
\begin{align*}\norm{X_{A'C}\ten Y_{AE}}{S_\infty( \mathsf{C}\otimes \mathsf{E},S_p(\mathsf{A}\otimes \mathsf{A'}))}
&\le
\norm{X_{A'C}}{S_\infty(\mathsf{C},  S_p(\mathsf{A' }) )}\norm{Y_{AE}}{S_\infty(\mathsf{E},  S_p(\mathsf{A} ) )}.
\end{align*}
The other direction follows from definition.

Step 4. Now we use the duality that for $\frac{1}{\al'}+\frac{1}{\al}=1$,
\[ S_1( \mathsf{C},  S_\al(\mathsf{A'} ) )^*
=S_\infty( \mathsf{C},  S_{\al'}(\mathsf{A'} ) ).\]
Then
\begin{align*}
\norm{\omega\ten \rho}{S_1(\mathsf{C}\otimes \mathsf{E},S_\al(\mathsf{A}\otimes \mathsf{A'}))}
=&\sup_{W_{AA'CE}}| \tr[W(\omega\ten \rho)]|\\
\ge&  \sup_{X_{A'C},Y_{AE}}| \tr[(X\ten Y)(\omega\ten \rho)]|\\
=&  \sup_{X_{A'C},Y_{AE}}| \tr [X\omega]|| \tr[Y\rho]|\\
=&\norm{\omega}{S_1(\mathsf{C},S_\al(\mathsf{A'}))}\norm{\rho}{S_1(\mathsf{E},S_\al(\mathsf{A}))},
\end{align*}
where the supremums above are over
\[  \norm{W_{A'CAE}}{S_\infty( \mathsf{C}\otimes \mathsf{E},S_{\al'}(\mathsf{A}\otimes \mathsf{A'}))}
=\norm{X_{A'C}}{S_\infty(\mathsf{C},  S_{\al'}(\mathsf{A' }) )}=\norm{Y_{AE}}{S_\infty(\mathsf{E},  S_{\al'}(\mathsf{A} ) )}=1\pl.   \]
That completes the proof.
\end{proof}

\end{document}